\documentclass[%
 aip,
 amsmath,amssymb,
reprint, onecolumn %%%%% Comment onecolumn to override one column format %%%%%%%
]{revtex4-1}

\usepackage{graphicx}% Include figure files
\usepackage{dcolumn}% Align table columns on decimal point
\usepackage{bm}% bold math
\usepackage{amsthm} % for theorem/proposition environments
\usepackage{algorithm}
\usepackage{algpseudocode}
\usepackage{longtable}

\usepackage[utf8]{inputenc}
\usepackage[T1]{fontenc}
\usepackage{mathptmx}
\usepackage{etoolbox}
\usepackage{amsthm,amsmath,amssymb}
\usepackage{xcolor} % for \textcolor
\usepackage[hidelinks]{hyperref}

\newcommand{\e}{\mathrm{e}}

\newtheorem{theorem}{Theorem}[section]
\newtheorem{lemma}{Lemma}[section]

\newtheorem{prop}{Proposition}[section]
\newtheorem{remark}{Remark}[section]

\makeatletter
\def\@email#1#2{%
 \endgroup
 \patchcmd{\titleblock@produce}
  {\frontmatter@RRAPformat}
  {\frontmatter@RRAPformat{\produce@RRAP{*#1\href{mailto:#2}{#2}}}\frontmatter@RRAPformat}
  {}{}
}%
\makeatother

\begin{document}

\preprint{AIP/123-QED}

\title[{\color{gray} Isotropic Coordinates for Generalized Schwarzschild-like Solutions}]{Isotropic Coordinates for Generalized Schwarzschild-like Solutions}
\author{Zeyu Zeng}
\affiliation{Department of Physics, The Grainger College of Engineering, University of Illinois Urbana-Champaign, Urbana, IL 61801, USA}
\affiliation{Department of Mathematics, College of Liberal Arts \& Sciences, University of Illinois Urbana-Champaign, Urbana, IL 61801, USA}
\author{Elena Kopteva}
\affiliation{Department of Physics, The Grainger College of Engineering, University of Illinois Urbana-Champaign, Urbana, IL 61801, USA}
\email{koptieva@illinois.edu}

\date{\today}% It is always \today, today,

\begin{abstract}
We consider a broad class of static, spherically symmetric generalized Schwarzschild-like solutions with multiple non-interacting anisotropic fluid sources and derive the coordinate transformation from Schwarzschild-like (curvature) to isotropic coordinates with conformally flat spatial slices. The isotropic form removes spatial-sector coordinate pathologies at the horizon, clarifies geometric quantities (e.g., ADM mass and curvature invariants), and enables the construction of well-posed initial data on $t=const$ hypersurfaces, suitable for the Hamiltonian and conformal formulations of numerical relativity and for perturbation theory. The backgrounds in isotropic coordinates we develop make it straightforward to separate environmental effects from intrinsic strong-gravity signals and meet the growing interest in non-vacuum black hole phenomenology across scattering, lensing, and waveform modeling. 
\end{abstract}

\maketitle
\tableofcontents

%############################################################
%#################### Introduction #########################
%###########################################################

%===============================================================
\section{Introduction}
\label{sec:intro}

Real astrophysical black holes do not live in vacuum.
Matter and fields around a compact object can change observable signals.
These effects can appear in ringdown spectra, gravitational wave phasing, and shadow structure.
This can happen even when the central object is close to Schwarzschild or Reissner--Nordstr\"om (RN) black hole.
Environmental systematics therefore matter for precision tests of strong gravity
\cite{Barausse2014environmental,Cardoso2019nature,Barack2019roadmap}.
Following Visser, we refer to such systems as dirty black holes \cite{Visser1992dirty}.

Cosmology adds another motivation.
Observations support accelerated expansion.
A black hole formed in such a universe interacts with nontrivial ambient energy components.
Kiselev-type metrics provide a common phenomenological description of a static black hole surrounded by an anisotropic effective fluid \cite{Kiselev2003QuintessenceBH}.
The associated stress energy is generically anisotropic and not that of a perfect fluid \cite{Visser2020KiselevNotPerfectFluid}.
General decompositions of the effective stress energy have been studied in detail \cite{Boonserm_2019phw}.
These backgrounds are widely used to study environmental signatures in quasinormal modes, shadows, and strong lensing \cite{Chen2005QNM,Abdujabbarov2017shadow,Younas2015lensing}.

Gravitational wave astronomy has made these questions more concrete.
Environmental structure can produce measurable phase and amplitude changes in waveforms.
Early work showed detectable dephasing from dark matter structures around a black hole \cite{Eda2013DMspike}.
Later studies developed discriminants between different mechanisms, including accretion and other surrounding media \cite{Cole2023environments,Speri2023accretion}.
Recent analyses have placed direct constraints on compact binary environments using LIGO-Virgo data \cite{CanevaSantoro2024first}.
These developments motivate tools that treat environmental effects in a systematic and coordinate transparent way.

For static spherically symmetric metrics, isotropic coordinates provide such a tool.
They make the spatial geometry conformally flat.
They also avoid coordinate blow ups at the horizon that appear in standard curvature coordinates.
This makes horizon regular initial data easier to construct.
It also simplifies many geometric quantities relevant to diagnostics and interpretation.
The same gauge is natural for the initial value formulation and the definition of global charges \cite{arnowitt2008dynamics}.

Isotropic coordinates are also tightly connected to numerical relativity practice.
Modern simulations often use conformal decompositions of the spatial metric in the BSSN framework \cite{shibata1995evolution,baumgarte1999shapiro}.
A common starting point is conformal flatness \cite{baumgarte2010numerical,alcubierre2008introduction}.
This choice underlies standard black hole initial data constructions, including Bowen-York data and puncture methods \cite{bowen1980timeasymmetric,brandt1997simple}.
It is implemented in widely used infrastructure such as the Einstein Toolkit and its puncture initial data solvers \cite{loffler2012einstein,zilhao2013introduction,ansorg2004single}.
Foundational reviews describe this pipeline in detail \cite{cook2000initial,Gourgoulhon2012}.
The first robust binary black hole merger simulations relied on this conformal infrastructure \cite{Campanelli2006,Baker2006}.

Beyond full numerical evolution, isotropic coordinates appear in several analytic frameworks.
They are the standard gauge used in parametrized post-Newtonian analyses \cite{Will2014confrontation,misner1973gravitation}.
They also play a central role in effective one body constructions \cite{Buonanno1999EOB}.
Close limit perturbation theory links conformally flat initial data to black hole perturbations \cite{Price1994closelimit}.
These links strengthen the case for having reliable isotropic maps for nonvacuum backgrounds.

Closed form isotropic transformations are known in a few classical cases.
Schwarzschild and RN solutions admit algebraic transformations \cite{misner1973gravitation,wald1984general,Weinberg1972,Kokkotas2020,payandeh2013isotropic,Casadio2013ChargedShells}.
The K\"ottler (Schwarzschild--de Sitter) solution is already more complex.
Its isotropic radius can be written in terms of elliptic functions \cite{Solanki:2021vqv}.
Related time dependent isotropic like constructions have also been developed \cite{dennison2017schwarzschild}.
Historical and systematic treatments of the K\"ottler family can be found in standard references \cite{kottler1918physikalischen,griffiths2009exact,stephani2003exact}.

Despite this literature, a constructive isotropic map for a broad class of Schwarzschild like metrics with multiple noninteracting anisotropic sources is not readily available.
Generic Kiselev type models are typically treated in curvature coordinates.
In many applications, isotropic coordinates are mentioned but not developed into a practical construction.
The complexity of the K\"ottler case suggests that general multisource models require new tools.

In this work we address this gap.
We provide a general construction of the isotropic radius for a broad class of static, spherically symmetric Schwarzschild-like metrics written in curvature coordinates.
We verify the construction on standard limiting cases and on representative nonvacuum examples.
We then develop a constructive local inversion procedure that recovers the areal radius as a series in the isotropic radius.
For practical computation, we analyze finite-order truncation and compare the series-based approach with direct numerical inversion.
Further convergence diagnostics, detailed error estimates, and implementation details are collected in the appendices.
These results are designed to provide a practical framework for working in isotropic coordinates in numerical relativity and in analytic studies of perturbations, including settings where regularity at horizons is important.

The paper is organized as follows.
Section~\ref{sec:setup} states the general setup, the parameter range in the generalized Kiselev-type ansatz, and the domain assumptions.
Section~\ref{sec:main} presents the general isotropic-coordinate transformation.
Section~\ref{sec:corollaries} develops corollaries and special cases, including classical limits and representative nonvacuum examples.
Section~\ref{SecLagInv} develops the constructive local inversion procedure for recovering $r(\rho)$ from $\rho(r)$ and discusses convergence of the inverse series.
Section~\ref{sec:error_analysis_double_truncation} treats finite-order truncation in practical computation and compares the series-based inversion strategy with numerical root finding.
Section~\ref{sec:discussion} discusses geometric interpretation, parameter dependence, applications, implementation, and scope.
Technical derivations, sharper convergence diagnostics, detailed error decomposition, numerical estimates, and pseudocode are collected in the appendices; principal notation is summarized in Appendix~\ref{app:notation}.

%############################################################
%#################### General Setup ########################
%###########################################################

%===============================================================
\section{Setup and Conventions}\label{sec:setup}

\subsection{General setup}

We work in geometric units with $c = 1$ and $8\pi G = 1$, and adopt the mostly-plus signature $(-,+,+,+)$. The Einstein equations follow from the Einstein-Hilbert action
\begin{equation}
S = \frac{1}{2} \int d^4x \, \sqrt{-g} \, R + S_{\text{matter}},
\end{equation}
and read
\begin{equation}\label{EE}
R^{\mu}{}_{\nu} - \frac{1}{2} \delta^{\mu}{}_{\nu} R = T^{\mu}{}_{\nu}.
\end{equation}

We consider static, spherically symmetric spacetimes with line element
\begin{equation}\label{ds}
    \mathrm{d}s^{2} = -f(r)\,\mathrm{d}t^{2} + f(r)^{-1}\,\mathrm{d}r^{2} + r^{2}\,\mathrm{d}\Omega^{2},
\end{equation}
where $\mathrm{d}\Omega^{2}$ is the metric on the unit 2-sphere and $f(r)$ is to be determined. When no confusion can arise, we will omit the argument and simply write $f$ instead of $f(r)$. For the metric ansatz \eqref{ds}, the Einstein equations are consistent with a diagonal stress--energy tensor of the form $T^{\mu}{}_{\nu}=\mathrm{diag}(-\varepsilon,\,p_r,\,p_\perp,\,p_\perp)$, where $\varepsilon$ is the energy density, and $p_r$ and $p_\perp$ are the radial and tangential pressures.

We assume that the surrounding matter is a superposition of noninteracting, anisotropic fluids, so that the total stress--energy tensor is a linear sum of individual ones (labeled by subscript $i$)
\begin{equation}\label{TsumTi}
    T^{\mu}{}_{\nu} = \sum_i T^{\mu}{}_{\nu (i)},
\end{equation}
each obeying its own conservation law $T^{\mu}{}_{\nu(i);\mu} = 0$.
For each component we adopt an effective barotropic relation in terms of the average pressure,
\begin{equation}\label{StateEq}
    \frac{p_{r(i)} + 2 p_{\perp(i)}}{3} = w_i \varepsilon_i,
\end{equation}
with constant equation-of-state parameter $w_i$.

With this setup, the metric function takes the generalized Kiselev-type form
\begin{equation}\label{f(r)}
f(r)=1-\frac{r_g}{r}-\sum_{i=1}^{N}\frac{K_i}{r^{\,3w_i+1}}.
\end{equation}
where $K_i$ and $w_i$ are, respectively, the dimensional normalization constant and the equation-of-state parameter of the $i$th matter source, and $r_g = 2M$ is the gravitational radius, where $M$ is the black hole mass parameter. In general, both $K_i$ and $w_i$ may take either sign, and $N$ is finite. The value $w_i=0$ is degenerate with the Schwarzschild sector, since then $3w_i+1=1$ and the corresponding contribution scales as $r^{-1}$; throughout, such terms are absorbed into a redefinition of $r_g$ rather than treated as separate matter components. 

Since $f(r)$ is a finite sum of distinct power-law terms, its positive-root structure is sharply constrained: if $m$ denotes the number of distinct exponents after merging coincident values, then the equation $f(r)=0$ has at most $m-1$ positive roots, counted with multiplicity. A proof of this root bound is given in Appendix~\ref{app:rootbound}. In particular, the finiteness of the set of positive roots of $f$ will be used below to specify the physically allowed exterior static region.

\subsection{Parameter range in the generalized Kiselev-type ansatz}

The admissible range of the constants $(K_i,w_i)$ in \eqref{f(r)} is not determined by \eqref{StateEq} alone. It depends on which physical requirements are imposed, and in particular on whether they are enforced \emph{componentwise}, for each individual contribution $T^{\mu}{}_{\nu(i)}$, or only for the \emph{total} stress--energy tensor $T^{\mu}{}_{\nu}$ in \eqref{TsumTi}.

A standard requirement is the dominant energy condition (DEC). For a Hawking--Ellis Type-I tensor $T^{\mu}{}_{\nu}=\mathrm{diag}(-\varepsilon,p_r,p_\perp,p_\perp),$
the DEC requires
\begin{equation}\label{eq:DEC}
\varepsilon \ge 0,
\qquad
\varepsilon \ge |p_r|,
\qquad
\varepsilon \ge |p_\perp|.
\end{equation}
See, e.g., \cite{HawkingEllis1973,wald1984general}.

If the DEC is imposed \emph{componentwise} on each $T^{\mu}{}_{\nu(i)}$, then the Kiselev principal-pressure relations lead to the conservative range
\begin{equation}\label{eq:param1}
-1 \le w_i \le \frac{1}{3},
\end{equation}
see \cite{Visser2020KiselevNotPerfectFluid}.

If instead the DEC is imposed only on the \emph{total} tensor $T^{\mu}{}_{\nu}$ on the domain of interest, then no universal radius-independent bound on each individual $w_i$ follows in general. Rather, admissibility becomes a property of the summed stress--energy tensor and may depend on radius when different components scale differently.

As noted above, the case $w_i=0$ is absorbed into the Schwarzschild sector, so we do not retain it as a separate component in what follows.
See~\cite{Kopteva:AnisotropicBH} for a detailed re-derivation of the Kiselev metric, a discussion of its properties, and a broader list of applications.

\subsection{Assumptions and Domain}\label{sec:domain}

Assume that $f$ admits an exterior static region, that is, a connected interval on which $f>0$ and whose lower endpoint is a positive root of $f$. We denote this lower endpoint by $r_+$ and refer to it as the event horizon. We denote by $r_\uparrow$ the upper endpoint of this exterior static region: $r_\uparrow=+\infty$ if $f>0$ extends to infinity (asymptotically flat/AdS case), and $r_\uparrow<+\infty$ if the static patch is bounded above by a cosmological horizon. Define
\begin{equation}
\mathcal{S}:=(r_+, r_\uparrow).
\end{equation}

We assume:
\begin{itemize}
\item[(A1)] $f$ is $C^\infty$ on $\mathcal{S}$, with $f>0$ on $\mathcal{S}$, and any finite endpoint of $\mathcal{S}$, namely $r_+$ and, when finite, $r_\uparrow$, is a simple zero of $f$.
\item[(A2)] $r$ is the areal radius on $\mathcal{S}$.
\item[(A3)] For every $r\in\mathcal{S}$, the function $x\mapsto \frac{1}{x}\bigl(\frac{1}{\sqrt{f(x)}}-1\bigr)$ is integrable on $(r_+,r]$.
\end{itemize}

%############################################################
%################# Main REsult #############################
%###########################################################

\section{General Isotropic-Coordinate Transformation}\label{sec:main}

\begin{theorem}[Main Result]\label{thm:main}
Under assumptions (A1)--(A3) and the domain definitions specified in \S\ref{sec:domain}, any static spherically symmetric metric of the form
\begin{equation}
 \mathrm{d}s^2 = -f(r)\mathrm{d}t^2 + f^{-1}(r)\mathrm{d}r^2 + r^2 \mathrm{d}\Omega^2
\end{equation}
\textit{admits an isotropic representation}
\begin{equation}
\mathrm{d}s^2 = -F(\rho)\mathrm{d}t^2 + \Phi(\rho)^4 \bigl(\mathrm{d}\rho^2 + \rho^2 \mathrm{d}\Omega^2\bigr),
\end{equation}
where the isotropic radius is defined implicitly by
\begin{equation}\label{expKiselev}
    \rho(r) = \rho_0 r \e^{I(r)},
\end{equation}
with
\begin{equation}\label{eq:I(r)}
I(r) = \int_{r_+}^{r} \frac{dx}{x}\!\left(\frac{1}{\sqrt{f(x)}} - 1\right),
\qquad r \in (r_+, r_\uparrow).
\end{equation}
Here \(F(\rho)=f(r(\rho))\) and \(\Phi(\rho)^2=r(\rho)/\rho\).

Moreover, \(\rho\) is strictly increasing on \(\mathcal{S}=(r_+,r_\uparrow)\), and therefore admits a unique inverse \(r(\rho)\) on \(\rho(\mathcal{S})\).
\end{theorem}

\begin{proof}\label{sec:proof}
We sketch a compact derivation of Eq.~\eqref{expKiselev} and the monotonicity
statement, postponing algebraic details to the Appendix. We start from a static, spherically symmetric line element in curvature
coordinates,
\begin{equation}
  \mathrm{d}s^2
  = -f(r)\,\mathrm{d}t^2 + f(r)^{-1}\,\mathrm{d}r^2 + r^2\,\mathrm{d}\Omega^2,
  \label{eq:metric-curvature-proof}
\end{equation}
and seek a purely radial coordinate transformation
\begin{equation}
  \Psi : (t,r,\theta,\phi)\mapsto (t,\rho,\theta,\phi),
  \qquad \rho=\rho(r),
\end{equation}
such that the metric takes the isotropic form
\begin{equation}
  \mathrm{d}s^2
  = -F(\rho)\,\mathrm{d}t^2
    + \Phi(\rho)^4\bigl(\mathrm{d}\rho^2+\rho^2\,\mathrm{d}\Omega^2\bigr).
  \label{eq:metric-isotropic-proof}
\end{equation}
Time and angular coordinates are unchanged, so the two metrics must agree term
by term.

Comparing the angular parts of Eqs.~\eqref{eq:metric-curvature-proof} and
\eqref{eq:metric-isotropic-proof} gives
\begin{equation}
  r^2\,\mathrm{d}\Omega^2
  = \Phi(\rho)^4\,\rho^2\,\mathrm{d}\Omega^2
  \quad\Longrightarrow\quad
  \Phi(\rho)^4 = \Bigl(\frac{r}{\rho}\Bigr)^2,
\end{equation}
so that
\begin{equation}
  \Phi(\rho)^2 = \frac{r}{\rho}.
  \label{eq:Phi-from-rho}
\end{equation}
In particular, once \(\rho(r)\) is known, \(\Phi(\rho)\) is determined algebraically. Matching the radial terms in
Eqs.~\eqref{eq:metric-curvature-proof}--\eqref{eq:metric-isotropic-proof},
using \eqref{eq:Phi-from-rho}, yields
\begin{equation}
  f(r)^{-1}\,\mathrm{d}r^2
  = \Phi(\rho)^4\,\mathrm{d}\rho^2
  = \Bigl(\frac{r}{\rho}\Bigr)^2 \mathrm{d}\rho^2.
\end{equation}
Taking square roots, with the positive sign on the exterior static domain \(\mathcal{S}\), we obtain
\begin{equation}
  \frac{\mathrm{d}r}{\mathrm{d}\rho}
  = \frac{r}{\rho}\sqrt{f(r)}
  \quad\Longleftrightarrow\quad
  \frac{1}{\rho}\frac{\mathrm{d}\rho}{\mathrm{d}r}
  = \frac{1}{r\sqrt{f(r)}}.
\end{equation}
Thus \(\rho(r)\) satisfies the first-order equation
\begin{equation}
  \frac{\mathrm{d}\rho}{\mathrm{d}r}
  = \frac{\rho}{r\sqrt{f(r)}},
  \qquad r\in\mathcal{S},
  \label{eq:drho-dr}
\end{equation}
which is the equation familiar from known explicit examples such as Schwarzschild, RN, and K\"ottler~\cite{misner1973gravitation,wald1984general,payandeh2013isotropic}. Equation~\eqref{eq:drho-dr} is separable:
\begin{equation}\label{eq:d_ln_rho}
  \frac{\mathrm{d}\rho}{\rho}
  = \frac{\mathrm{d}r}{r\sqrt{f(r)}}.
\end{equation}
Fix the lower integration limit at \(r_+\) and set
\(\rho(r_+)=\rho_+>0\). Since $r_+$ is a simple zero of $f$, the integrand in \eqref{eq:I(r)} is integrable at the lower endpoint, so the normalization at $r_+$ is well defined. Integrating from \(r_+\) to \(r\) gives
\begin{equation}
  \ln\frac{\rho(r)}{\rho_+}
  = \int_{r_+}^{r}\frac{\mathrm{d}x}{x\sqrt{f(x)}}.
  \label{eq:rho-integral-1}
\end{equation}
To separate off the reference logarithmic term corresponding to the flat radial scaling $\rho\propto r$, add and subtract $1/x$ in the integrand:
\begin{equation}
  \frac{1}{x\sqrt{f(x)}}
  = \frac{1}{x}\Bigl(\frac{1}{\sqrt{f(x)}}-1\Bigr)
    + \frac{1}{x}.
\end{equation}
Substituting this into \eqref{eq:rho-integral-1} yields
\begin{equation}
  \ln\frac{\rho(r)}{\rho_+}
  = \int_{r_+}^{r}\frac{\mathrm{d}x}{x}\Bigl(\frac{1}{\sqrt{f(x)}}-1\Bigr)
    + \int_{r_+}^{r}\frac{\mathrm{d}x}{x}
  = \int_{r_+}^{r}\frac{\mathrm{d}x}{x}\Bigl(\frac{1}{\sqrt{f(x)}}-1\Bigr)
    + \ln\frac{r}{r_+}.
\end{equation}
Exponentiating and absorbing \(r_+\) into a redefined constant
\(\rho_0=\rho_+ /r_+\) gives
\begin{equation}
  \rho(r)
  = \rho_0\,r\,
    \exp\!\Biggl[
      \int_{r_+}^{r}\frac{\mathrm{d}x}{x}
      \Bigl(\frac{1}{\sqrt{f(x)}}-1\Bigr)
    \Biggr],
\end{equation}
which is precisely Eq.~\eqref{expKiselev}.

Finally, on \(\mathcal{S}\) one has \(r>0\), \(f(r)>0\), and \(\rho(r)>0\), so Eq.~\eqref{eq:drho-dr} implies
\begin{equation}
\frac{\mathrm{d}\rho}{\mathrm{d}r}
=
\frac{\rho}{r\sqrt{f(r)}}>0.
\end{equation}
Therefore \(\rho\) is strictly increasing on \(\mathcal{S}\), and hence admits a unique inverse \(r(\rho)\) on \(\rho(\mathcal{S})\).
\end{proof}

\begin{remark}
Equations~\eqref{expKiselev} and~\eqref{eq:I(r)} are the primary object of the theorem: we produce \(\rho(r)\) in closed integral form. 
Section~\ref{SecLagInv} gives a constructive Lagrange inversion method to obtain \(r(\rho)\) when needed.
\end{remark}

\begin{prop}[Multi-index expansion for generalized Kiselev-type backgrounds]
\label{prop:multiindex}
If, in addition, the metric function is of generalized Kiselev type \eqref{f(r)},
\begin{equation}
f(r)=1-\frac{r_g}{r}-\sum_{i=1}^{N}\frac{K_i}{r^{\,3w_i+1}}, \quad w_i \neq 0
\end{equation}
then, writing
\begin{equation}
s_i:=3w_i+1,\qquad
\xi_0(r)=\frac{r_g}{r},\qquad
\xi_i(r)=\frac{K_i}{r^{s_i}}\ \ (i=1,\dots,N),
\end{equation}
so that
\begin{equation}
f(r)=1-\xi(r),\qquad \xi(r)=\sum_{i=0}^{N}\xi_i(r),
\end{equation}
the integral $I(r)$ in Eq.~\eqref{eq:I(r)} admits the following formal expansion
\begin{equation}
I(r)
= \sum_{\substack{\mathbf{m}\neq \mathbf 0}}
C(\mathbf m)\, r_g^{m_0}\prod_{i=1}^{N} K_i^{m_i}
\int_{r_+}^{r} \mathrm{d} x\, x^{-1-\alpha(\mathbf m)},
\qquad r\in(r_+,r_\uparrow),
\label{eq:I-series-general}
\end{equation}
where \,  \(\mathbf m=(m_0,m_1,\dots,m_N)\in \mathbb N_0^{\,N+1}\), \,
\(|\mathbf m|=m_0+m_1+\cdots+m_N\ge 1\), \,  and
\begin{equation}
C(\mathbf m)
=
\frac{(2|\mathbf m|)!}{4^{|\mathbf m|} (|\mathbf m|!)^2}
\frac{|\mathbf m|!}{\prod_{i=0}^{N} (m_i !)},
\qquad
\alpha(\mathbf m)=m_0+\sum_{i=1}^{N} m_i s_i.
\end{equation}
\end{prop}

\begin{proof}
To derive the formal multi-index expansion of Eq.~\eqref{eq:I(r)}, we start from
\begin{equation}
I(r)
=
\int_{r_+}^r \frac{\mathrm{d} x}{x}\left(\frac{1}{\sqrt{f(x)}}-1\right).
\end{equation}
Using \(f(x)=1-\xi(x)\), Taylor's expansion gives
\begin{equation}\label{expansion1}
    \frac{1}{\sqrt{f(x)}} = \frac{1}{\sqrt{1 - \xi(x)}} = \sum_{k = 0}^\infty \binom{2k}{k} \frac{\xi(x)^k}{4^k}
    = \sum_{k = 0}^{\infty} \frac{(2k)!}{4^k (k!)^2}\xi(x)^k.
\end{equation}
As a generalization of Newton's binomial theorem, we expand \(\xi(x)^k\) using the multinomial theorem:
\begin{equation}
    (x_0 + x_1 + \cdots + x_N)^k
    =
    \sum_{\substack{m_0 + m_1 + \cdots + m_N = k \\ m_0, m_1, \dots, m_N \geq 0}}
    \frac{k!}{m_0! \cdots m_N!} x_0^{m_0} \cdots x_N^{m_N}.
\end{equation}
Thus, introducing the multi-index \(\mathbf m=(m_0,m_1,\dots,m_N)\), \(|\mathbf m|=m_0+m_1+\cdots+m_N=k\), we write
\begin{equation}
\xi(x)^k
=
(\xi_0 + \cdots + \xi_N)^k
=
\sum_{|\mathbf m|=k}
\frac{|\mathbf m|!}{\prod_{i=0}^N (m_i !)}
\prod_{i = 0}^N \xi_i(x)^{m_i}.
\end{equation}
By the definition of \(\xi_i(x)\), we also have
\begin{equation}
\prod_{i=0}^N \xi_i(x)^{m_i}
=
r_g^{m_0}\,x^{-m_0-\sum_{i=1}^N m_i s_i}\,\prod_{i=1}^N K_i^{m_i}
=
r_g^{m_0}\,\prod_{i=1}^N K_i^{m_i}\,x^{-\alpha(\mathbf m)}.
\end{equation}
Substituting the result above into Eq.~\eqref{expansion1}, we arrive at
\begin{equation}
\frac{1}{\sqrt{f(x)}} - 1
=
\sum_{\mathbf{m}\neq \mathbf 0}
\frac{(2|\mathbf m|)!}{4^{|\mathbf m|} (|\mathbf m|!)^2}
\frac{|\mathbf m|!}{\prod_{i=0}^N (m_i !)}
\,r_g^{m_0}\prod_{i=1}^N K_i^{m_i}\,
x^{-\alpha(\mathbf m)}.
\end{equation}
Defining
\begin{equation}
C(\mathbf m)
=
\frac{(2|\mathbf m|)!}{4^{|\mathbf m|} (|\mathbf m|!)^2}
\frac{|\mathbf m|!}{\prod_{i=0}^N (m_i !)},
\qquad
\alpha(\mathbf m)=m_0+\sum_{i=1}^{N} m_i s_i,
\end{equation}
and substituting into Eq.~\eqref{eq:I(r)}, the integral can be written as
\begin{equation}
I(r)
=
\sum_{\substack{\mathbf{m}\neq \mathbf 0}}
C(\mathbf m)\, r_g^{m_0}\prod_{i=1}^{N} K_i^{m_i}
\int_{r_+}^{r} \mathrm{d} x\, x^{-1-\alpha(\mathbf m)},
\qquad r\in(r_+,r_\uparrow),
\end{equation}
which is precisely Eq.~\eqref{eq:I-series-general}.
\end{proof}

%############################################################
%####### Corollaries and Special Cases #####################
%###########################################################

\section{Corollaries and Special Cases}\label{sec:corollaries}

%===============================================================
\subsection{Schwarzschild limit}\label{sec:schwarzschild_example}

As a first nontrivial check of Theorem~\ref{thm:main}, we consider the Schwarzschild metric, for which $r_+=r_g$, $r_\uparrow=\infty$, and hence $\mathcal{S}=(r_g,\infty)$.
\begin{equation}
  \mathrm{d}s^2
  = -\Bigl(1 - \frac{r_g}{r}\Bigr)\mathrm{d}t^2
    + \Bigl(1 - \frac{r_g}{r}\Bigr)^{-1}\mathrm{d}r^2
    + r^2 \mathrm{d}\Omega^2.
  \label{eq:Sch-curvature}
\end{equation}
The standard isotropic form of the Schwarzschild metric is well known~\cite{Kokkotas2020},
\begin{equation}
  \mathrm{d}s^2
  = -\left(\frac{1 - \dfrac{r_g}{4\rho}}
                 {1 + \dfrac{r_g}{4\rho}}\right)^2 \mathrm{d}t^2
    + \left(1 + \frac{r_g}{4\rho}\right)^4
      \bigl(\mathrm{d}\rho^2+\rho^2\mathrm{d}\Omega^2\bigr),
  \label{eq:Sch-isotropic-standard}
\end{equation}
with the curvature and isotropic radii related by
\begin{equation}
  r(\rho) = \rho\left(1 + \frac{r_g}{4\rho}\right)^2,
  \qquad
  \rho \in \Bigl(\frac{r_g}{4},\infty\Bigr).
  \label{eq:Sch-rho-standard}
\end{equation}
We now show that this relation follows from Eqs.~\eqref{expKiselev} and~\eqref{eq:I-series-general}.

For Schwarzschild we have $f(r) = 1 - r_g/r$.
In the notation of Proposition~\ref{prop:multiindex}, there is a single nonzero term $\xi_0(r) = r_g/r$, and all $K_i$ vanish. The multi-index reduces to $\mathbf{m}=(k,0,\dots,0)$ with $|\mathbf{m}|=k$, and
\begin{equation}
  C(\mathbf{m}) = C_k
  := \frac{(2k)!}{4^k (k!)^2}, \qquad \alpha(\mathbf{m}) = k.
\end{equation}
Expanding the integrand in Eq.~\eqref{eq:I(r)} by the binomial series, we write on the Schwarzschild exterior domain $\mathcal{S}=(r_g,\infty)$
\begin{equation}
  I_{\mathrm{Sch}}(r)
  = \int_{r_g}^{r}\frac{\mathrm{d}x}{x}
    \Bigl(\frac{1}{\sqrt{1-r_g/x}} - 1\Bigr)
  = \sum_{k=1}^{\infty} C_k r_g^{k}
    \int_{r_g}^{r} x^{-1-k}\,\mathrm{d}x .
  \label{eq:Sch-I-start}
\end{equation}
For $x>r_g$ the binomial series is valid termwise, and the resulting closed form extends continuously to the lower endpoint $x=r_g$.
The elementary integral needed here is
\begin{equation}
  \int_{r_g}^{r} x^{-1-k}\,\mathrm{d}x
  = -\frac{1}{k}\bigl(r^{-k} - r_g^{-k}\bigr),\quad k\ge1.
\end{equation}
Therefore,
\begin{equation}
  I_{\mathrm{Sch}}(r)
  = \sum_{k=1}^{\infty} C_k r_g^{k}
    \left[-\frac{1}{k}\bigl(r^{-k} - r_g^{-k}\bigr)\right]
  = \sum_{k=1}^{\infty}\frac{C_k}{k}\Bigl(1 - (r_g/r)^k\Bigr).
  \label{eq:Sch-I-series}
\end{equation}
Introduce $t:=r_g/r\in(0,1)$ and define
\begin{equation}
  S(t)
  := \sum_{k=1}^{\infty}\frac{C_k}{k} t^k
   = \sum_{k=1}^{\infty}\binom{2k}{k}\frac{t^k}{4^k k}.
  \label{eq:Sch-S-def}
\end{equation}
Then Eq.~\eqref{eq:Sch-I-series} can be written compactly as
\begin{equation}
  I_{\mathrm{Sch}}(r)
  = S(1) - S\!\left(\frac{r_g}{r}\right).
  \label{eq:Sch-I-S}
\end{equation}

\begin{prop}\label{prop:series-identity}
  For $t\in(-1,1)$ the series
  \begin{equation}
    S(t)
    = \sum_{k=1}^{\infty}\binom{2k}{k}\frac{t^k}{4^k k}
  \end{equation}
  converges and admits the closed form
  \begin{equation}\label{series-1}
    S(t)
    = -2\ln\!\left(\frac{1+\sqrt{1-t}}{2}\right).
  \end{equation}
  Moreover, the series also converges at $t=1$, with
  \begin{equation}
    S(1)=2\ln 2.
  \end{equation}
\end{prop}

\begin{proof}
  The ratio test shows that the series converges at least for
  $t\in(-1,1)$.
  Using
  \begin{equation}
    \frac{1}{k} = \int_0^1 x^{k-1}\,\mathrm{d}x,
  \end{equation}
  we rewrite $S(t)$ as
  \begin{equation}
    S(t)
    = \int_0^1
      \left[
        \sum_{k=1}^{\infty}
        \binom{2k}{k}\left(\frac{tx}{4}\right)^k
      \right]\frac{\mathrm{d}x}{x}
    = \int_0^1 \bigl[(1-tx)^{-1/2}-1\bigr]x^{-1}\,\mathrm{d}x,
  \end{equation}
  where we used the binomial series for $(1-z)^{-1/2}$ and exchanged
  sum and integral inside the convergence domain. Differentiating with
  respect to $t$ under the integral sign gives
  \begin{equation}
    S'(t)
    = \frac{1}{2}\int_0^1 (1-tx)^{-3/2}\,\mathrm{d}x
    = \frac{1}{t}\bigl[(1-t)^{-1/2}-1\bigr]
    = \frac{1-\sqrt{1-t}}{t\sqrt{1-t}}.
  \end{equation}
  On the other hand, if we define
  \begin{equation}
    F(t) := -2\ln\left(\frac{1+\sqrt{1-t}}{2}\right),
  \end{equation}
  a straightforward calculation shows
  \begin{equation}
    F'(t)
    = \frac{1-\sqrt{1-t}}{t\sqrt{1-t}} = S'(t).
  \end{equation}
  Since $S(0)=F(0)=0$, it follows that $S(t)=F(t)$ on $(-1,1)$,
  establishing~\eqref{series-1}. Since all coefficients are nonnegative,
  $S(t)$ is increasing on $[0,1)$. Taking the limit $t\to1^{-}$ in
  \eqref{series-1} therefore gives
  \begin{equation}
    S(1)=\lim_{t\to1^-}S(t)=2\ln 2.
  \end{equation}
\end{proof}

Using Proposition~\ref{prop:series-identity}, we have $S(1)=2\ln 2$. For general $t=r_g/r$ we have
\begin{equation}
  S\!\left(\frac{r_g}{r}\right)
  = -2\ln\!\left(\frac{1+\sqrt{1-r_g/r}}{2}\right).
\end{equation}
Substituting into~\eqref{eq:Sch-I-S} yields the closed form
\begin{equation}
  I_{\mathrm{Sch}}(r)
  = 2\ln 2
    + 2\ln\!\left(\frac{1+\sqrt{1-r_g/r}}{2}\right)
  = 2\ln\!\bigl(1+\sqrt{1-r_g/r}\bigr).
  \label{eq:Sch-I-closed}
\end{equation}
The general transformation formula~\eqref{expKiselev} then gives
\begin{equation}
  \rho(r)
  = \rho_0\,r\,\exp\bigl[I_{\mathrm{Sch}}(r)\bigr]
  = \rho_0\,r\bigl(1+\sqrt{1-r_g/r}\bigr)^2
  = \rho_0\bigl(\sqrt{r}+\sqrt{r-r_g}\bigr)^2.
  \label{eq:Sch-rho-from-theorem}
\end{equation}
Setting $\rho_0=\tfrac{1}{4}$ and inverting~\eqref{eq:Sch-rho-from-theorem} gives
\begin{equation}
  r(\rho) = \rho\left(1 + \frac{r_g}{4\rho}\right)^2,
\end{equation}
which is exactly the standard relation~\eqref{eq:Sch-rho-standard}
underlying the isotropic form~\eqref{eq:Sch-isotropic-standard}.

Thus, the standard Schwarzschild isotropic coordinates are recovered from Theorem~\ref{thm:main} together with the series expansion \eqref{eq:I-series-general}.

%===================================================================================

\subsection{Reissner--Nordstr\"om solution}\label{sec:rn_example}

As a second nontrivial check of Theorem~\ref{thm:main} we turn to the
RN metric,
\begin{equation}
  \mathrm{d}s^2_{\mathrm{RN}}
  = -\left(1-\frac{r_g}{r}+\frac{Q^2}{r^2}\right)\mathrm{d}t^2
    +\left(1-\frac{r_g}{r}+\frac{Q^2}{r^2}\right)^{-1}\mathrm{d}r^2
    +r^2\mathrm{d}\Omega^2, \qquad 0<Q^2<\frac{r_g^2}{4}.
  \label{eq:RN-curvature}
\end{equation}
The metric function can be written as
\begin{equation}
  f(r)=1-\frac{r_g}{r}+\frac{Q^2}{r^2}
  =\frac{(r-r_+)(r-r_-)}{r^2},
  \qquad
  r_\pm=\frac{1}{2}\Bigl(r_g\pm\sqrt{r_g^2-4Q^2}\Bigr),
  \label{eq:RN-f-horizons}
\end{equation}
and on the physical exterior region one has
\begin{equation*}
  r_\uparrow=\infty,
  \qquad
  \mathcal{S}=(r_+,\infty),
\end{equation*}
with $f(r)>0$ for all $r\in\mathcal{S}$.

The isotropic form of the RN metric is also standard in the literature~\cite{Kokkotas2020,Casadio2013ChargedShells}:
\begin{equation}
  \mathrm{d}s^2
  = -\left(
      \frac{1-\dfrac{r_g^2-4Q^2}{16\rho^2}}
           {1+\dfrac{r_g}{2\rho}+\dfrac{r_g^2-4Q^2}{16\rho^2}}
    \right)^2\mathrm{d}t^2
    +\left(
       1+\frac{r_g}{2\rho}+\frac{r_g^2-4Q^2}{16\rho^2}
     \right)^2\bigl(\mathrm{d}\rho^2+\rho^2\mathrm{d}\Omega^2\bigr),
  \label{eq:RN-isotropic-standard}
\end{equation}
with the curvature and isotropic radii related by 
\begin{equation}
  r(\rho)
  =\rho\left(
       1+\frac{r_g}{2\rho}+\frac{r_g^2-4Q^2}{16\rho^2}
     \right),
  \label{eq:RN-rho-standard}
\end{equation}
which maps the outer horizon $r_+$ to the finite isotropic radius
\begin{equation*}
  \rho_+ = \frac{1}{4}\sqrt{r_g^2-4Q^2}.
\end{equation*}

We now show that this relation follows from Eq.~\eqref{expKiselev} together with the RN specialization of the formal expansion \eqref{eq:I-series-general}.

\subsubsection{Series representation and consistency with the general integral}

For the RN metric we write
\begin{equation}
  f(r)=1-\xi(r),
  \qquad
  \xi(r)=\frac{r_g}{r}-\frac{Q^2}{r^2},
  \label{eq:RN-xi-def}
\end{equation}
so that $0\le\xi(r)<1$ on the exterior domain $r>r_+$. The binomial
expansion gives
\begin{equation}
  \frac{1}{\sqrt{f(r)}}
  =\frac{1}{\sqrt{1-\xi(r)}}
  =\sum_{k=0}^{\infty} C_k\,\xi(r)^k,
  \qquad
  C_k:=\frac{1}{4^k}\binom{2k}{k}
      =\frac{(2k)!}{4^k(k!)^2}.
  \label{eq:RN-binomial}
\end{equation}
In the notation of Proposition~\ref{prop:multiindex}, the RN metric corresponds
to a single non-Schwarzschild term with
\begin{equation*}
  s_1=2,
  \qquad
  K_1=-Q^2,
\end{equation*}
so the multi-index reduces to
$\mathbf{m}=(m_0,m_1)\in\mathbb{N}_0^2\setminus\{(0,0)\}$. Specializing Eq.~\eqref{eq:I-series-general} to the RN case and evaluating the resulting elementary integrals, we obtain the series representation
\begin{equation}
  I_{\mathrm{RN}}(r)
  = \sum_{(m_0,m_1)\neq(0,0)}
    \frac{(2(m_0+m_1))!}{4^{\,m_0+m_1}(m_0+m_1)!}\,
    \frac{r_g^{m_0}(-Q^2)^{m_1}}{m_0!\,m_1!\,(m_0+2m_1)}
    \left(r_+^{-(m_0+2m_1)}-r^{-(m_0+2m_1)}\right),
  \label{eq:RN-I-series}
\end{equation}
which converges for each fixed $r>r_+$.

\begin{lemma}
  Fix $R_0>r_+$. For $r\in [R_0,\infty)$ one has
  \begin{equation}
    r\,\frac{\mathrm{d}I_{\mathrm{RN}}}{\mathrm{d}r}(r)
    =\sum_{k\ge1} C_k\,\xi(r)^k,
    \label{eq:RN-rdI-sum}
  \end{equation}
  and the series on the right is absolutely and uniformly convergent on $[R_0,\infty)$.
\end{lemma}

\begin{proof}
  Differentiating Eq.~\eqref{eq:RN-I-series}, the constant
  term $r_+^{-(m_0+2m_1)}$ drops out and the factor $(m_0+2m_1)$
  cancels, so with $k=m_0+m_1$ we obtain
  \begin{equation}
    r\,\frac{\mathrm{d}I_{\mathrm{RN}}}{\mathrm{d}r}(r)
    =
    \sum_{(m_0,m_1)\neq(0,0)}
    \frac{(2k)!}{4^k k!}\,
    \frac{r_g^{m_0}(-Q^2)^{m_1}}{m_0!\,m_1!}\,
    r^{-(m_0+2m_1)}.
  \end{equation}
  Grouping terms with fixed $k=m_0+m_1$ and using
  \begin{equation*}
    \sum_{m_0+m_1=k}
    \frac{a^{m_0}b^{m_1}}{m_0!\,m_1!}
    =\frac{(a+b)^k}{k!},
  \end{equation*}
  with
  \begin{equation*}
    a=\frac{r_g}{r},
    \qquad
    b=-\frac{Q^2}{r^2},
  \end{equation*}
  gives
  \begin{equation}
    r\,\frac{\mathrm{d}I_{\mathrm{RN}}}{\mathrm{d}r}(r)
    =\sum_{k\ge1}\frac{(2k)!}{4^k(k!)^2}
      \left(\frac{r_g}{r}-\frac{Q^2}{r^2}\right)^k
    =\sum_{k\ge1} C_k\,\xi(r)^k.
  \end{equation}
  Since $0\le \xi(r)\le \xi(R_0)<1$ for $r\in [R_0,\infty)$, the series
  on the right is absolutely and uniformly convergent on $[R_0,\infty)$.
\end{proof}

\begin{prop}
  For $r>r_+$,
  \begin{equation}
    r\,\frac{\mathrm{d}I_{\mathrm{RN}}}{\mathrm{d}r}(r)
    =f(r)^{-1/2}-1.
    \label{eq:RN-main-identity}
  \end{equation}
\end{prop}

\begin{proof}
  The generating function of the coefficients $C_k$ is
  \begin{equation}
    \sum_{k=0}^{\infty}C_k x^k=(1-x)^{-1/2}.
  \end{equation}
  Since Eq.~\eqref{eq:RN-rdI-sum} starts at $k=1$,
  \begin{equation}
    r\,\frac{\mathrm{d}I_{\mathrm{RN}}}{\mathrm{d}r}(r)
    =\sum_{k\ge1} C_k\,\xi(r)^k
    =(1-\xi(r))^{-1/2}-1
    =f(r)^{-1/2}-1,
  \end{equation}
  which is Eq.~\eqref{eq:RN-main-identity}.
\end{proof}

Comparing Eq.~\eqref{eq:RN-main-identity} with the defining relation
\begin{equation}
  \frac{\mathrm{d}I(r)}{\mathrm{d}r}
  =\frac{1}{r}\left(\frac{1}{\sqrt{f(r)}}-1\right),
\end{equation}
we see that the series \eqref{eq:RN-I-series}  coincides with the integral \eqref{eq:I(r)} in the RN case. Since Eq.~\eqref{eq:RN-I-series} is normalized so that $I_{\mathrm{RN}}(r_+)=0$, it follows that
\begin{equation}
  I_{\mathrm{RN}}(r)=\int_{r_+}^{r}\frac{\mathrm{d}x}{x}
  \left(\frac{1}{\sqrt{f(x)}}-1\right).
\end{equation}
Thus, the isotropic radius given by Theorem~\ref{thm:main},
\begin{equation}
  \rho(r)=\rho_0\,r\,\exp\!\bigl(I_{\mathrm{RN}}(r)\bigr),
  \qquad r\in(r_+,\infty),
  \label{eq:RN-rho-from-theorem}
\end{equation}
solves
\begin{equation}
  \frac{\mathrm{d}\rho}{\mathrm{d}r}
  =\frac{\rho}{r\sqrt{f(r)}}
\end{equation}
on the RN exterior domain.

It is straightforward to check that the explicit relation
\begin{equation}
  \rho(r)
  = \frac{1}{2}\left(r - \frac{r_g}{2} + \sqrt{r^2-r_gr+Q^2}\right),
  \qquad r>r_+,
  \label{eq:RN-rho-explicit}
\end{equation}
obeys the same differential equation and satisfies the horizon normalization
\begin{equation}
  \rho(r_+)=\rho_+=\frac{1}{4}\sqrt{r_g^2-4Q^2}.
\end{equation}
Inverting Eq.~\eqref{eq:RN-rho-explicit} gives
Eq.~\eqref{eq:RN-rho-standard}. Substituting this relation into the
isotropic form given in Theorem~\ref{thm:main} then reproduces the
standard RN isotropic metric \eqref{eq:RN-isotropic-standard}.

In summary, the RN solution provides a second check that the formal expansion
\eqref{eq:I-series-general}, together with Eq.~\eqref{expKiselev}, correctly
reproduces the familiar RN isotropic coordinates for a nontrivial two-parameter
family of metrics.

%===============================================================

\subsection{K\"ottler metric}\label{sec:Kottler_example}

As a third check of Theorem~\ref{thm:main} we consider the
K\"ottler solution
\begin{equation}
  \mathrm{d}s^2
  = -\left(1-\frac{r_g}{r}-\frac{\Lambda}{3}r^2\right)\mathrm{d}t^2
    +\left(1-\frac{r_g}{r}-\frac{\Lambda}{3}r^2\right)^{-1}\mathrm{d}r^2
    + r^2\mathrm{d}\Omega^2,
  \label{eq:K\"ottler-curvature}
\end{equation}
with metric function
\begin{equation}
  f(r)
  = 1-\frac{r_g}{r}-\frac{\Lambda}{3}r^2
  = -\frac{\Lambda}{3r}\,(r-r_1)(r-r_2)(r-r_3),
  \label{eq:K\"ottler-f-factor}
\end{equation}
where $r_1>r_2>0>r_3$ are the three real roots of
$-\tfrac{\Lambda}{3}r^3 + r - r_g = 0$.
We identify the horizons as
\begin{equation}
  r_{\uparrow}:=r_1,\qquad r_+:=r_2,
\end{equation}
so that for $\Lambda>0$ the static region is
$\mathcal{S}=(r_+,r_{\uparrow})=(r_2,r_1)$.

\subsubsection{Isotropic coordinates via elliptic functions}

As shown in Ref.~\cite{Solanki:2021vqv} (and rederived in
Appendix~\ref{app:Kottler-isotropic}), the isotropic radius $\rho$ for the
K\"ottler metric can be expressed in terms of the incomplete elliptic
integral of the first kind,
\begin{equation}
  u = F(\phi,k)
  := \int_0^{\phi}\frac{\mathrm{d}\theta}
   {\sqrt{1-k^2\sin^2\theta}},
  \qquad
  k^2=\frac{-\,r_3(r_1-r_2)}{\,r_1(r_2-r_3)\,}\in(0,1),
  \label{eq:K\"ottler-F-def}
\end{equation}
with amplitude
\begin{equation}
  \phi(r) = \arcsin\sqrt{\frac{\,r_1(r-r_2)\,}{\,r(r_1-r_2)\,}}.
  \label{eq:K\"ottler-phi-def}
\end{equation}
Choosing $\rho(r_2)=\rho_+$ one finds
\begin{equation}
  \ln\frac{\rho(r)}{\rho_+}
  = \sqrt{\frac{3}{\Lambda}}\,
    \frac{2}{\sqrt{\,r_1(r_2-r_3)\,}}\,
    F\!\bigl(\phi(r),k\bigr),
  \qquad r\in(r_2,r_1).
  \label{eq:K\"ottler-rho-F}
\end{equation}
Equivalently,
\begin{equation}
  \rho(r)
  = \rho_+\,
    \exp\!\left[
      \sqrt{\frac{3}{\Lambda}}\,
      \frac{2}{\sqrt{\,r_1(r_2-r_3)\,}}\,
      F\!\bigl(\phi(r),k\bigr)
    \right].
  \label{eq:K\"ottler-rho-explicit}
\end{equation}

Inverting is most conveniently expressed in terms of the Jacobi
elliptic sine $\mathrm{sn}(u,k)$, defined by
\(
  u = F(\phi,k),\; \mathrm{sn}(u,k) = \sin\phi
\).
Introducing
\begin{equation}
  u(\rho)
  := \frac{\sqrt{\,r_1(r_2-r_3)\,}}{2}\,
     \sqrt{\frac{\Lambda}{3}}\,
     \ln\frac{\rho}{\rho_+},
  \qquad
  s(\rho):=\mathrm{sn}^2\!\bigl(u(\rho),k\bigr),
  \label{eq:K\"ottler-u-sn-def}
\end{equation}
and solving
\(
  \tfrac{r_1(r-r_2)}{r(r_1-r_2)} = s(\rho)
\)
for $r$ gives the Jacobi form
\begin{equation}
  r(\rho)
  = \frac{\,r_1 r_2\,}
    {\,r_1 - (r_1-r_2)\,\mathrm{sn}^2\!\bigl(u(\rho),k\bigr)\,},
  \qquad \rho\in(\rho_+,\rho_{\uparrow}),
  \label{eq:K\"ottler-r-of-rho}
\end{equation}
where $\rho_+=\rho(r_2)$ and $\rho_{\uparrow}=\rho(r_1)$ are the isotropic
images of the two horizons. Substituting $r(\rho)$ into
\eqref{eq:K\"ottler-curvature} leads to the K\"ottler metric in isotropic
(conformally flat) form
\begin{equation}
  \mathrm{d}s^2_{\mathrm{Kottler}}
  = -\Bigl[1-\frac{r_g}{r(\rho)}-\frac{\Lambda}{3}r(\rho)^2\Bigr]\mathrm{d}t^2
    +\Bigl(\frac{r(\rho)}{\rho}\Bigr)^2
     \bigl(\mathrm{d}\rho^2+\rho^2\mathrm{d}\Omega^2\bigr).
  \label{eq:K\"ottler-iso-metric}
\end{equation}

\subsubsection{Consistency with the general transformation formula}

For the K\"ottler metric we write
\begin{equation}
  f(r)
  = 1-\frac{r_g}{r}-\frac{\Lambda}{3}r^2
  = 1-\xi(r),
  \qquad
  \xi(r):=\frac{r_g}{r}+K_1 r^2,
  \qquad K_1:=\frac{\Lambda}{3},
  \label{eq:K\"ottler-xi-again}
\end{equation}
so that on the static interval $r\in(r_2,r_1)$ one has $0<\xi(r)<1$ and
the factorisation \eqref{eq:K\"ottler-f-factor},
\begin{equation}
  1-\xi(r)
  = 1-\frac{r_g}{r}-K_1 r^2
  = -\frac{\Lambda}{3r}\,(r-r_1)(r-r_2)(r-r_3),
\end{equation}
with $r_i$ as in Eq.~\eqref{eq:K\"ottler-f-factor}.

Specialising the general integral definition~\eqref{eq:I(r)} to K\"ottler, we
set
\begin{equation}
  I_{\mathrm{K}}(r)
  := \int_{r_2}^{r}\frac{\mathrm{d}x}{x}
     \left(\frac{1}{\sqrt{f(x)}}-1\right),
  \qquad r\in(r_2,r_1),
  \label{eq:IK-def}
\end{equation}
so that Theorem~\ref{thm:main} gives
\begin{equation}
  \rho(r)
  = \frac{\rho_+}{r_2}\,r\,\exp I_{\mathrm{K}}(r),
  \qquad r\in(r_2,r_1),
  \label{eq:rho-from-IK}
\end{equation}
with $\rho_+=\rho(r_2)$.

On the other hand, the explicit isotropic radius obtained in the
previous subsection is
\begin{equation}
  \ln\frac{\rho(r)}{\rho_+}
  = C_\Lambda\,F\!\bigl(\phi(r),k\bigr),
  \qquad
  C_\Lambda:=\sqrt{\frac{3}{\Lambda}}\,
              \frac{2}{\sqrt{\,r_1(r_2-r_3)\,}},
  \label{eq:K\"ottler-rho-elliptic-again}
\end{equation}
where $k$ and $\phi(r)$ are exactly as in
Eqs.~\eqref{eq:K\"ottler-F-def}--\eqref{eq:K\"ottler-phi-def}. Define
\begin{equation}
  J(r)
  := C_\Lambda\,F\!\bigl(\phi(r),k\bigr)
  = \ln\frac{\rho(r)}{\rho_+},
  \qquad r\in(r_2,r_1).
  \label{eq:J-def}
\end{equation}

\paragraph{Differential check.}
Using $\partial_\phi F(\phi,k)=(1-k^2\sin^2\phi)^{-1/2}$, the chain rule
gives
\begin{equation}
  J'(r)
  = \frac{\mathrm{d}}{\mathrm{d}r}
    \bigl[C_\Lambda F(\phi(r),k)\bigr]
  = C_\Lambda\,\frac{\phi'(r)}
    {\sqrt{\,1-k^2\sin^2\phi(r)\,}}.
\end{equation}
A short but straightforward computation using the definitions of
$\phi(r)$, $k$ and the cubic factorisation of $f(r)$ shows that
\begin{equation}
  J'(r)
  = \frac{1}{r\sqrt{f(r)}},
  \qquad r\in(r_2,r_1),
  \label{eq:Jprime-equals}
\end{equation}
with the normalisation $C_\Lambda$ chosen exactly so that the
right-hand side matches the K\"ottler integrand. (The details are given
in Appendix~\ref{app:Kottler-isotropic}.)

Comparing \eqref{eq:Jprime-equals} with \eqref{eq:IK-def}, we see that
\begin{equation}
  \frac{\mathrm{d}}{\mathrm{d}r}
  \Bigl[J(r)-\ln\frac{r}{r_2}\Bigr]
  = \frac{1}{r\sqrt{f(r)}}-\frac{1}{r}
  = \frac{\mathrm{d}I_{\mathrm{K}}}{\mathrm{d}r}(r),
\end{equation}
and both sides vanish at $r=r_2$:
\begin{equation}
  J(r_2)=0,\quad \ln\frac{r_2}{r_2}=0,\quad I_{\mathrm{K}}(r_2)=0.
\end{equation}
Hence
\begin{equation}
  I_{\mathrm{K}}(r)
  = J(r)-\ln\frac{r}{r_2}
  = C_\Lambda\,F\!\bigl(\phi(r),k\bigr)
    - \ln\frac{r}{r_2},
  \qquad r\in(r_2,r_1).
  \label{eq:IK-equals-elliptic}
\end{equation}
Substituting \eqref{eq:IK-equals-elliptic} into the general
transformation \eqref{eq:rho-from-IK} recovers
the elliptic expression \eqref{eq:K\"ottler-rho-elliptic-again} for $\rho(r)$.

\paragraph{Series viewpoint.}
If we expand
\begin{equation}
  \frac{1}{\sqrt{f(r)}}
  = \frac{1}{\sqrt{1-\xi(r)}}
  = \sum_{n=0}^{\infty} a_n\,\xi(r)^n,
  \qquad
  a_n = \frac{(2n)!}{4^n (n!)^2},
\end{equation}
and insert
\(
  \xi(r)=\dfrac{r_g}{r}+K_1 r^2,
\)
then by termwise integration the definition \eqref{eq:IK-def} reproduces
the K\"ottler specialization of the formal multi-index expansion in
Proposition~\ref{prop:multiindex}, namely Eq.~\eqref{eq:I-series-general}
with $N=1$ and $K_1=\Lambda/3$.

Conversely, the elliptic representation admits a corresponding local series expansion; the coefficient comparison with the formal K\"ottler specialization of Proposition~\ref{prop:multiindex} is carried out in Appendix~\ref{app:Kottler-isotropic}. In the main text, it suffices to note that both representations satisfy the same first-order equation and the same normalization at the base point $r=r_+=r_2$, and hence coincide on $(r_2,r_1)$.

%===============================================================

\subsection{Kiselev-Type Background Example: Reissner--Nordstr\"om plus one Kiselev term}\label{Ex.RN+K}

As a concrete application of Theorem~\ref{thm:main}, consider a charged
black hole surrounded by a single Kiselev-type fluid. In this case, the metric function is written in the form
\begin{equation}
  f(r) = 1 - \xi(r),
  \qquad
  \xi(r)
  = \frac{r_g}{r} - \frac{Q^2}{r^2} + \frac{K_1}{r^{s}},
  \label{eq:RNK-xi}
\end{equation}
where $r_g=2M$, with $M$ the mass parameter, $Q$ the charge parameter, 
$K_1$ the Kiselev amplitude, and $s = 3w+1, \,\, w\neq 0.$
 
As before, $r_+$ denotes the event horizon, and $r_\uparrow$ the outer
boundary of the static region, with $r_\uparrow=\infty$ if no cosmological horizon is
present. Thus $f(r)>0$ for $r\in(r_+,r_\uparrow)$.

Specializing Theorem~\ref{thm:main} to this background, the isotropic transformation is
\begin{equation}
  \rho(r) = \rho_0\,r\,\exp I(r),
  \qquad
  I(r) = \int_{r_+}^{r}\frac{\mathrm{d}x}{x}
    \left(\frac{1}{\sqrt{f(x)}}-1\right),
  \qquad r\in(r_+,r_\uparrow),
  \label{eq:RNK-I-def}
\end{equation}
where $\rho_0>0$ is the scale constant from Theorem~\ref{thm:main} 
$\bigl(\rho_0=\rho(r_+)/r_+\bigr)$.

To apply Proposition~\ref{prop:multiindex} to this background, we decompose $\xi$ into its
Schwarzschild, charge, and Kiselev contributions:
\begin{equation}
  \xi(r)=\xi_0(r)+\xi_1(r)+\xi_2(r),
  \qquad
  \xi_0(r)=\frac{r_g}{r},\quad
  \xi_1(r)=-\frac{Q^2}{r^2},\quad
  \xi_2(r)=\frac{K_1}{r^{s}}.
\end{equation}
Then
\begin{equation}
  \frac{1}{\sqrt{f(r)}}
  = \frac{1}{\sqrt{1-\xi(r)}}
  = \sum_{k=0}^{\infty} C_k\,\xi(r)^k,
  \qquad
  C_k:=\frac{(2k)!}{4^k(k!)^2},
  \label{eq:RNK-binomial}
\end{equation}
and the multinomial expansion gives, for each $k\ge1$,
\begin{equation}
  \xi(r)^k
  = \bigl(\xi_0(r)+\xi_1(r)+\xi_2(r)\bigr)^k
  = \sum_{\substack{m_0,m_1,m_2\ge0\\ m_0+m_1+m_2=k}}
    \frac{k!}{m_0!\,m_1!\,m_2!}\,
    \xi_0(r)^{m_0}\xi_1(r)^{m_1}\xi_2(r)^{m_2}.
\end{equation}
Writing $\mathbf{m}=(m_0,m_1,m_2)$ and $|\mathbf{m}|=k$, this becomes
\begin{equation}
  \xi(r)^k
  = \sum_{|\mathbf{m}|=k}
    \frac{k!}{m_0!\,m_1!\,m_2!}\,
    \left(\frac{r_g}{r}\right)^{m_0}
    \left(-\frac{Q^2}{r^2}\right)^{m_1}
    \left(\frac{K_1}{r^{s}}\right)^{m_2}
  = \sum_{|\mathbf{m}|=k}
    \frac{k!}{m_0!\,m_1!\,m_2!}\,
    r_g^{m_0}(-Q^2)^{m_1}K_1^{m_2}\,
    r^{-\alpha(\mathbf{m})},
\end{equation}
with exponent
\begin{equation}
  \alpha(\mathbf{m})
  := m_0 + 2m_1 + s\,m_2.
\end{equation}
Substituting into Eq.~\eqref{eq:I(r)} and using Eq.~\eqref{eq:I-series-general} from Proposition~\ref{prop:multiindex}, we write

\begin{equation}
  I(r)
  = \sum_{\mathbf{m}\neq\mathbf{0}}
    C(\mathbf{m})\,r_g^{m_0}(-Q^2)^{m_1}K_1^{m_2}\,
    \mathcal{J}_{\alpha(\mathbf{m})}(r;r_+),
  \qquad r\in(r_+,r_\uparrow),
  \label{eq:RNK-I-series}
\end{equation}
where
\begin{equation}
  C(\mathbf{m})
  := C_k\,\frac{k!}{m_0!\,m_1!\,m_2!}
   = \frac{(2k)!}{4^k (k!)^2}\,
     \frac{k!}{m_0!\,m_1!\,m_2!},
  \qquad k=|\mathbf{m}|,
\end{equation}
and the elementary radial integrals are
\begin{equation}
  \mathcal{J}_{\alpha}(r;r_+)
  := \int_{r_+}^{r} x^{-1-\alpha}\,\mathrm{d}x
  =
  \begin{cases}
    -\dfrac{r^{-\alpha}-r_+^{-\alpha}}{\alpha},
      & \alpha \neq 0,\\[1.2ex]
    \ln\!\dfrac{r}{r_+},
      & \alpha = 0.
  \end{cases}
  \label{eq:RNK-Jalpha}
\end{equation}

Therefore, in the RN--Kiselev case, the isotropic radius takes the form
\begin{equation}
  \rho(r)
  = \rho_0\, r\, e^{I(r)},
  \qquad r\in(r_+,r_\uparrow),
\end{equation}
with $I(r)$ given by the specialized multi-index expansion \eqref{eq:RNK-I-series}.

%===============================================================

%############################################################
%####### Constructive Lagrange Inversion ###################
%###########################################################

\section{Constructive Lagrange inversion: recovering $r(\rho)$ from $\rho(r)$}\label{SecLagInv}

For the isotropic map produced by Theorem~\ref{thm:main}, the isotropic radius is given implicitly by
\begin{equation}
  \rho(r)=\rho_0\,r\,e^{I(r)},\qquad
  I(r)=\int_{r_+}^{r}\frac{\mathrm{d}x}{x}\left(\frac{1}{\sqrt{f(x)}}-1\right),
  \label{eq:rho-I-def}
\end{equation}
and, in generalized Kiselev-type backgrounds, a closed-form global inverse $r(\rho)$ is generally unavailable. This section gives a constructive method for recovering $r(\rho)$ from $\rho(r)$: (i) we build a \emph{local} analytic inverse as a convergent power series via Lagrange--B{\"u}rmann inversion~\cite{AbramowitzStegun1972-LagrangeExpansion}, (ii) we describe what controls the radius of convergence in terms of the nearest obstruction to analytic continuation, and (iii) we outline a practical series-reversion scheme for generating high-order coefficients.
%, and (iv) we analyze the total error from double truncation and compare series-based methods with numerical root-finding algorithms.

On the real static domain, Theorem~\ref{thm:main} implies that $\rho(r)$ is strictly increasing on $\mathcal{S}$, and hence that the real inverse $r(\rho)$ exists and is unique on $\rho(\mathcal{S})$. For the generalized Kiselev-type metrics considered here, the only additional issue is local complex analyticity, which is needed for Lagrange--B{\"u}rmann inversion; we now show that, in the present class of metrics, this local analyticity holds at every regular point of the static region.

%===============================================================
%===============================================================

\subsection{Local inversion via Lagrange-B{\"u}rmann}\label{subsec:LagBurm}

We recall a standard form of the inversion theorem, adapted to the variables and notation of Theorem~\ref{thm:main}.
\begin{theorem}[Lagrange--B{\"u}rmann inversion]\label{LagrangeInversion}
Let $\Psi$ be analytic in a neighborhood of $r=r_\ast$, and assume that
$\Psi'(r_\ast)\neq 0$. Define
\begin{equation}
  \rho_\ast:=\Psi(r_\ast).
\end{equation}
Then $\Psi$ admits a local analytic inverse $r(\rho)$ near $\rho=\rho_\ast$, and
\begin{equation}
  r(\rho)
  =
  r_\ast+\sum_{n=1}^{\infty}\frac{(\rho-\rho_\ast)^n}{n!}
  \left[
    \frac{\mathrm{d}^{\,n-1}}{\mathrm{d}r^{\,n-1}}
    \left(\frac{r-r_\ast}{\Psi(r)-\Psi(r_\ast)}\right)^{\!n}
  \right]_{r=r_\ast}.
\end{equation}
The series has a nonzero radius of convergence.
\end{theorem}

%===============================================================

\subsubsection{Application to the isotropic map}

Fix a base point $r_\ast\in\mathcal{S}$ and define $\rho_\ast:=\rho(r_\ast)$. For the generalized Kiselev-type metrics, $f(r)$ is locally analytic at every regular point of the static region. Since $f(r_\ast)>0$, the function $1/\sqrt{f(r)}$ is analytic in a neighborhood of $r_\ast$. We now consider the local isotropic map
\begin{equation}
  \Psi(r):=\rho(r).
\end{equation}

The differential relation \eqref{eq:d_ln_rho} obtained in the proof of Theorem~\ref{thm:main} is
\begin{equation}
  \frac{\Psi'(r)}{\Psi(r)}=\frac{1}{r\sqrt{f(r)}},
  \qquad r\in\mathcal{S},
  \label{eq:Psi-log-deriv}
\end{equation}
which is independent of the constant $\rho_0$ in \eqref{eq:rho-I-def}. Define
\begin{equation}
  g(r):=\frac{\Psi'(r)}{\Psi(r)}=\frac{1}{r\sqrt{f(r)}}.
  \label{eq:g-def}
\end{equation}
Then $g$ is analytic near $r_\ast$, and so is $\Psi$, since they are related by the ODE
\begin{equation}\label{eq:PsiODE}
\Psi'(r)=g(r)\Psi(r).
\end{equation}
Moreover, since $\Psi(r_\ast)=\rho_\ast>0$,
\begin{equation}
\Psi'(r_\ast)=\Psi(r_\ast)\,g(r_\ast)
=\frac{\Psi(r_\ast)}{r_\ast\sqrt{f(r_\ast)}}\neq 0.
\end{equation}
Therefore Theorem~\ref{LagrangeInversion} applies to the map $\Psi(r)=\rho(r)$. Thus the local inverse $r(\rho)$ exists, is analytic in a neighborhood of $\rho_\ast$, and is given by
\begin{equation}\label{eq:KiselevInv}
  r(\rho)
  =
  r_\ast+\sum_{n=1}^{\infty}\frac{(\rho-\rho_\ast)^n}{n!}
  \left[
    \frac{\mathrm{d}^{\,n-1}}{\mathrm{d}r^{\,n-1}}
    \left(\frac{r-r_\ast}{\Psi(r)-\Psi(r_\ast)}\right)^{\!n}
  \right]_{r=r_\ast}.
\end{equation}
The radius of convergence is determined by the nearest obstruction to analytic continuation of
$r(\rho)$ in the complex $\rho$-plane; this is discussed in \S\ref{sec:InvConvergence}.

%===============================================================

\subsubsection{Low-order coefficients}

Let $y:=\rho-\rho_\ast$. Expanding \eqref{eq:KiselevInv} gives
\begin{equation}\label{eq:rLowOrder}
  r(\rho)=r_\ast+c_1 y+c_2 y^2+c_3 y^3+O(y^4),
\end{equation}
where
\begin{align}
  c_1&=\frac{1}{\Psi'(r_\ast)},\label{eq:c1}\\[2mm]
  c_2&=-\frac{\Psi''(r_\ast)}{2\,\Psi'(r_\ast)^{3}},\label{eq:c2}\\[2mm]
  c_3&=\frac{3\,\Psi''(r_\ast)^{2}-\Psi'(r_\ast)\Psi'''(r_\ast)}
            {6\,\Psi'(r_\ast)^{5}}.\label{eq:c3}
\end{align}
Since $\Psi'=g\Psi$, successive derivatives can be expressed in terms of $g$:
\begin{equation}\label{eq:PsiDerivs}
  \Psi''=(g^2+g')\Psi,\qquad
  \Psi'''=(g^3+3gg'+g'')\Psi.
\end{equation}
Thus all coefficients $c_n$ can be written in terms of $g$ and its derivatives evaluated at $r_\ast$.
It is convenient to introduce
\begin{equation}\label{eq:calAdef}
  \mathcal{A}(r):=\frac{1}{r}+\frac{f'(r)}{2f(r)},
  \qquad
  \mathcal{A}'(r)=-\frac{1}{r^2}+\frac{f''(r)f(r)-(f'(r))^2}{2f(r)^2}.
\end{equation}
Because $g=r^{-1}f^{-1/2}$, one has the compact identities
\begin{equation}\label{eq:gDerivsA}
  \frac{g'}{g}= -\mathcal{A},\qquad
  g'=-\mathcal{A}\,g,\qquad
  g''=(\mathcal{A}^2-\mathcal{A}')\,g,
\end{equation}
and therefore $\Psi^{(k)}(r_\ast)$ can be expressed explicitly through derivatives of $f$ at $r_\ast$.

For generalized Kiselev-type metrics, one may also express the derivatives of $g$ through the multi-index expansion of $f(r)^{-1/2}$ used earlier; this yields explicit coefficient formulas in parameter regimes where that expansion is valid. Since this machinery is auxiliary to the local inversion itself, we record it in Appendix~\ref{app:LagrangeAux}.

%===============================================================
%===============================================================

\subsection{Convergence and rate of the inverse series}\label{sec:InvConvergence}

Theorem~\ref{LagrangeInversion} guarantees that the Lagrange--B{\"u}rmann series
\eqref{eq:KiselevInv} has a nonzero radius of convergence at any base point $r_\ast$ with
$\Psi'(r_\ast)\neq 0$. The detailed complex-analytic description of the convergence radius and its obstructions is recorded in Appendix~\ref{app:InvConvergence}. Here we only note one point that is especially relevant for applications: in the nonextremal case covered by assumption~(A1), the horizon values themselves do not obstruct analyticity of the inverse. More precisely, the next lemma shows that a simple horizon gives $\mathrm{d}r/\mathrm{d}\rho=0$ at $\rho=\rho_h$, while $r(\rho)$ remains analytic there.

\begin{lemma}[Local behavior at a simple horizon]\label{lem:SimpleHorizon}
Let $r_h>0$ be a simple root of $f$ that is a finite endpoint of a static interval, i.e. $f(r_h)=0$ and $f'(r_h)\neq 0$, and define $\rho_h$ as the one-sided limit of $\rho(r)$ as $r\to r_h$ from within that static region. 
Then, as $r\to r_h$ from within the static region,
\begin{equation}
\rho(r)=\rho_h\Bigl(1\pm\gamma_h\sqrt{|r-r_h|}+O(|r-r_h|)\Bigr),
\qquad
\gamma_h=\frac{2}{r_h\sqrt{|f'(r_h)|}}>0,
\label{eq:rhoPuiseux}
\end{equation}
where the plus sign corresponds to the lower endpoint $r_h=r_+$ and the minus sign to the upper endpoint $r_h=r_\uparrow$. Consequently, with the same sign convention,
\begin{equation}
r-r_h
=
\pm \frac{1}{\gamma_h^{2}\rho_h^{2}}(\rho-\rho_h)^2
+O\!\left((\rho-\rho_h)^3\right).
\label{eq:rHorizonTaylor}
\end{equation}
In particular, the inverse is analytic in $\rho$ at $\rho=\rho_h$.
\end{lemma}

\begin{proof}
For a simple root at a finite endpoint of a static interval, one has
\begin{equation}
f(r)=f'(r_h)(r-r_h)+O\!\left((r-r_h)^2\right),
\end{equation}
and therefore, along the static side,
\begin{equation}\label{eq:g-horizon-asympt}
\frac{1}{r\sqrt{f(r)}}
=
\frac{1}{r_h\sqrt{|f'(r_h)|}}\,|r-r_h|^{-1/2}
+O(1).
\end{equation}
Using Eq.~\eqref{eq:Psi-log-deriv}, with $\Psi=\rho$, together with Eq.~\eqref{eq:g-horizon-asympt}, integration from $r_h$ to $r$ along the static side gives
\begin{equation}
\ln \frac{\rho(r)}{\rho_h}
=
\pm \gamma_h\sqrt{|r-r_h|}
+O(|r-r_h|),
\end{equation}
where the plus sign corresponds to the lower endpoint $r_h=r_+$ and the minus sign to the upper endpoint $r_h=r_\uparrow$. Exponentiating yields \eqref{eq:rhoPuiseux}, and inversion gives \eqref{eq:rHorizonTaylor}.
\end{proof}

\begin{remark}[Extremal horizons and reduced analyticity]\label{rem:Extremal}
This situation lies outside assumption (A1), which requires simple zeros at finite endpoints of
$\mathcal{S}$. If $f$ has a double root at $r_h$ (extremality), then $\sqrt{f(r)}\sim \kappa |r-r_h|$ on the relevant real side, and the
integrand $1/(r\sqrt{f})$ has a simple pole, producing a logarithmic term in $\ln\rho$. In that
case $\rho$ behaves locally as a power of $|r-r_h|$, so that the inverse behaves locally as
$|r-r_h|\sim C\,\rho^{1/\beta}$ after a suitable normalization. Unless $1/\beta$ is a positive integer, this introduces a
branch point for $r(\rho)$ and may reduce the convergence radius in \eqref{eq:Rdef}. In what
follows we restrict attention to the nonextremal case covered by Theorem~\ref{thm:main}, where the local inverse is analytic at the
horizon values.
\end{remark}

Standard coefficient-based diagnostics for the convergence radius, including root and ratio tests, Domb--Sykes analysis, and truncation-error estimates, are useful when implementing the inverse series at high order. Since these tools are auxiliary to the construction itself, we collect them in Appendix~\ref{app:InvConvergence}. We now turn to the practical generation of high-order coefficients.
%===============================================================
%===============================================================

\subsection{Practical generation of high-order coefficients}

For convergence diagnostics one needs coefficients $c_n$ at moderately high order (typically
$N\sim 30$--$100$). Direct evaluation of \eqref{eq:KiselevInv} via $(n\!-\!1)$-th derivatives becomes
rapidly unstable at high $n$. A numerically robust alternative is to compute the inverse series by
\emph{series reversion}, using only truncated power-series algebra. Introduce shifted variables
\begin{equation}
x:=r-r_\ast,\qquad y:=\rho-\rho_\ast,
\end{equation}
and expand the forward map as $y=A(x)=\sum_{k=1}^{N}a_k x^k$ with $a_1=\rho'(r_\ast)\neq 0$.
Then the inverse has the form $x=B(y)=\sum_{n=1}^{N}c_n y^n$ and is obtained by enforcing the
composition identity $A(B(y))\equiv y$ order by order. A convenient implementation proceeds as:

\begin{enumerate}
\item Expand the metric function in a neighborhood of $r_\ast$,
\begin{equation}
f(r_\ast+x)=\sum_{k=0}^{N} f_k x^k+O(x^{N+1}),
\end{equation}
using the local generalized binomial expansion of $(r_\ast+x)^{-s_i}$ about $x=0$.
\item Compute the truncated series for $1/\sqrt{f(r_\ast+x)}$ by Newton iteration in
the ring of truncated series, and form
\begin{equation}
g(r_\ast+x)=\frac{1}{(r_\ast+x)\sqrt{f(r_\ast+x)}}=\sum_{k=0}^{N} g_k x^k+O(x^{N+1}).
\end{equation}
\item Solve $\mathrm{d}\rho(r_\ast+x)/\mathrm{d}x=g(r_\ast+x)\,\rho(r_\ast+x)$ in series form. Writing
$\rho(r_\ast+x)=\sum_{k=0}^{N}\rho_k x^k$ with $\rho_0=\rho_\ast$, coefficient matching yields
the recursion
\begin{equation}
(k+1)\rho_{k+1}=\sum_{j=0}^{k} g_j\,\rho_{k-j},\qquad k=0,1,\dots,N-1.
\label{eq:rhoRecursion}
\end{equation}
Then $y=\rho-\rho_\ast=\sum_{k=1}^{N} a_k x^k$ with $a_k=\rho_k$.
\item Perform series reversion to obtain $x=\sum_{n=1}^{N}c_n y^n$ by enforcing
$A(B(y))\equiv y$ to order $N$. This produces $c_n$ recursively from $a_1,\dots,a_n$ without
high-order differentiation.
\end{enumerate}

The coefficients computed by this procedure agree with those obtained from Lagrange inversion
\eqref{eq:KiselevInv}. They can then be used in the coefficient-based convergence diagnostics summarized in Appendix~\ref{app:InvConvergence}, \eqref{eq:RratioRoot}--\eqref{eq:DombSykes}, for Schwarzschild, RN, K\"ottler, and generalized Kiselev-type metrics, including multi-component cases.

%##########################################
%####### Error analysis ###################
%##########################################

\section{Truncation in practical computation and comparison with numerical inversion}
\label{sec:error_analysis_double_truncation}

% \subsection{Truncation Error Decomposition}

In practical applications, the isotropic transformation must be computed with finite resources, requiring truncation of the formal multi-index expansion \eqref{eq:I-series-general} at $|\mathbf m| \leq k_{\max}$ and, if needed, the Lagrange inverse \eqref{eq:KiselevInv} at order $N$. 
This section summarizes the resulting truncation error and compares the series-based construction with Newton--Raphson numerical inversion.

It is convenient to separate the error of the doubly truncated approximation $r^{(N,k_{\max})}(\rho)$ into the schematic contributions
\begin{equation}
\left|r^{(N,k_{\max})}(\rho)-r(\rho)\right|
\leq
\varepsilon_{I}^{(k_{\max})}
+
\varepsilon_{L}^{(N)}
+
\varepsilon_{\mathrm{coupling}},
\label{eq:total_error}
\end{equation}
where $\varepsilon_{I}^{(k_{\max})}$ is the truncation error in the forward map $I(r)$, $\varepsilon_{L}^{(N)}$ is the truncation error in the inverse series, and $\varepsilon_{\mathrm{coupling}}$ accounts for the propagation of the forward truncation into the inverse coefficients. A more detailed discussion of these contributions, together with illustrative balancing estimates for $N$ and $k_{\max}$, is given in Appendix~\ref{app:double_truncation_error}.

Representative order-of-magnitude estimates for Schwarzschild, Reissner--Nordstr\"om, and K\"ottler benchmark cases are collected in Appendix~\ref{app:error_estimates_examples}, together with a summary table of illustrative truncation choices and resulting error scales.

%===============================================================

\subsection{Practical comparison with numerical inversion}
\label{sec:ComparisonAlg}

An alternative to series inversion is to solve the implicit equation $\rho(r) - \rho_{\text{target}} = 0$ numerically for $r$ using Newton--Raphson iteration~\cite{press2007numerical,burden2015numerical}:
\begin{equation}
r_{n+1} = r_n - \frac{\rho(r_n) - \rho_{\text{target}}}{\rho'(r_n)}, \quad n = 0, 1, 2, \ldots,
\label{eq:newton_raphson}
\end{equation}
where
\begin{equation}
\rho'(r)=\frac{\rho(r)}{r\sqrt{f(r)}},
\end{equation}
in agreement with Eq.~\eqref{eq:drho-dr}. Newton--Raphson exhibits \textit{quadratic convergence}: if $e_n := |r_n-r_{\text{true}}|$ is the error at iteration $n$, then
\begin{equation}
e_{n+1} \leq C \, e_n^2,
\label{eq:quadratic_convergence}
\end{equation}
provided the initial guess $r_0$ is sufficiently close to $r_{\text{true}}$ and $\rho'(r_{\text{true}}) \neq 0$. In favorable regimes, only a few iterations may suffice to reach the accuracy permitted by the evaluation of the forward map, and may approach machine precision ($\sim 10^{-15}$) when $\rho(r)$ is computed numerically to comparable accuracy.

Table~\ref{tab:method_comparison} compares the Lagrange series and Newton--Raphson approaches across several performance metrics.

\begin{table}[ht]
\centering
\caption{Comparison of Lagrange-series inversion and Newton--Raphson root-finding for computing $r(\rho)$ from the implicit relation $\rho(r)$. Costs are quoted per evaluation of $r(\rho)$ after any one-time setup. For the Lagrange method, the setup consists of computing the coefficients $c_n$ from the truncated forward expansion. For Newton--Raphson, each iteration requires evaluation of the forward map $\rho(r)$ (and, when used explicitly, its derivative). The table is meant as a practical comparison of usage regimes, not a sharp complexity analysis.}
\label{tab:method_comparison}
\begin{tabular}{lcc}
\hline\hline
\textbf{Criterion} & \textbf{Lagrange Series} & \textbf{Newton--Raphson} \\
\hline
Setup cost & Compute $c_n$ once & Optional table / initial guess \\
Evaluation cost & Horner: $O(N)$ & $3$--$5$ forward-map iterations \\
Typical accuracy & $\sim 10^{-6}$ to $10^{-10}$ & Truncation-limited, or near machine precision \\
Analytic derivatives & Yes (explicit from the series) & Available via implicit differentiation \\
Global robustness & Local (valid near $\rho_\ast$) & Depends on initial guess and target point \\
Best use case & Regular grids, symbolic work & Scattered points, maximum accuracy \\
\hline\hline
\end{tabular}
\end{table}

\paragraph{Practical recommendations.}
The choice between methods depends on the application:\\

\textbf{Lagrange series is preferred} when:
\begin{enumerate}
\item $r(\rho)$ must be evaluated at many points on a regular grid (e.g., Cartesian initial data for numerical relativity codes~\cite{loffler2012einstein,ansorg2004single,cook2000initial}),
\item explicit analytic expressions for derivatives $\partial^j r / \partial \rho^j$ are needed for computing Christoffel symbols or curvature invariants,
\item moderate accuracy ($10^{-6}$ to $10^{-10}$) suffices,
\item symbolic or algebraic manipulation is desired.
\end{enumerate}
In this regime, the one-time cost of computing coefficients $c_n$ is amortized over many evaluations, and Horner's method~\cite{knuth1997art} provides very fast polynomial evaluation.\\

\textbf{Newton--Raphson is preferred} when:
\begin{enumerate}
\item very high accuracy (machine precision) is required,
\item $r(\rho)$ is needed at only a few scattered evaluation points (e.g., adaptive mesh refinement),
\item the convergence radius $R$ of the Lagrange series is small, limiting its domain of validity,
\item a reliable initial guess is available, for example from a nearby grid point or from a truncated Lagrange approximation.
\end{enumerate}
Each Newton--Raphson call requires $3$--$5$ evaluations of $\rho(r)$ via the truncated forward map, but this overhead is acceptable for sparse queries.\\

\paragraph{Hybrid strategy for numerical relativity.}
In practice, a \textit{hybrid approach} often yields the best performance:
\begin{enumerate}
\item Compute the forward map $I(r)$ via the truncated multi-index series \eqref{eq:I-series-general} with $k_{\max} = 20$--$30$, aiming for $\varepsilon_{\text{I}} \sim 10^{-9}$ to $10^{-10}$ on the chosen region. see Appendix~\ref{app:error_estimates_examples} for representative examples underlying these values.
\item For setting up initial data on a uniform Cartesian grid $(x, y, z)$ in the Einstein Toolkit~\cite{loffler2012einstein,zilhao2013introduction}, use local Lagrange inverses
\begin{equation*}
r(\rho)=\sum_{n=0}^N c_n(\rho-\rho_\ast)^n
\end{equation*}
with $N$ chosen in proportion to $k_{\max}$ on suitable overlapping $\rho$-patches to obtain $r(\rho)$ at grid points $\rho=\sqrt{x^2+y^2+z^2}$; see Appendix~\ref{app:error_estimates_examples} for representative choices. From this, construct the conformal factor $\Phi(\rho)^2=r(\rho)/\rho$ and, in numerical-relativity notation, the lapse $\alpha(\rho)=\sqrt{F(\rho)}=\sqrt{f(r(\rho))}$.
\item For post-processing or analysis requiring very high accuracy (e.g., constraint violation checks, horizon location), apply Newton--Raphson with the Lagrange result as the initial guess.
\end{enumerate}
This strategy leverages the strengths of both methods: the Lagrange series provides fast, analytic evaluation across the computational domain through a suitable patchwise construction, while Newton--Raphson refines accuracy where needed~\cite{baumgarte2010numerical,alcubierre2008introduction}.

Detailed convergence diagnostics and benchmark-based validation are collected in Appendix~\ref{app:series_diagnostics}, including Cauchy--Hadamard and Domb--Sykes analyses, Richardson-style checks across truncation orders, and comparison with exact benchmark cases.

%############################################################
%######### Discussion and Interpretation ###################
%###########################################################

\section{Discussion and Interpretation}\label{sec:discussion}

\subsection{Geometric meaning of the isotropic representation}

Theorem~\ref{thm:main} gives a constructive map from curvature (areal-radius) coordinates to an isotropic radius $\rho$ such that the spatial metric on $t=\mathrm{const}$ slices is conformally flat:
\begin{equation}
\gamma_{ij}\,\mathrm{d}x^i\mathrm{d}x^j
=\Phi(\rho)^4\bigl(\mathrm{d}\rho^2+\rho^2\mathrm{d}\Omega^2\bigr),
\qquad
\Phi(\rho)^2=\frac{r(\rho)}{\rho}.
\end{equation}
The areal radius is then given implicitly by $r(\rho)$. The conformal factor $\Phi$ determines the intrinsic geometry of the slice. This form matches standard conformal decompositions used in numerical relativity and in conformally flat initial data constructions \cite{shibata1995evolution,baumgarte1999shapiro,alcubierre2008introduction,baumgarte2010numerical,cook2000initial,Gourgoulhon2012}. Once $r(\rho)$ is known (analytically in special cases, or by the inversion method of Sec.~\ref{SecLagInv}), spatial quantities needed in diagnostics---Christoffel symbols, 3-curvature, and constraint expressions---follow from derivatives of $\Phi(\rho)$.

In curvature coordinates the map is defined by
\begin{equation}
\frac{1}{\rho}\frac{\mathrm{d}\rho}{\mathrm{d}r}=\frac{1}{r\sqrt{f(r)}},
\end{equation}
with integrated form given by Eq.~\eqref{expKiselev}. Up to the scale $\rho_0$, this is the unique radial coordinate that yields a conformally flat spatial metric for the class \eqref{ds}. The functional $I(r)$ in Eq.~\eqref{eq:I(r)} measures the departure from $\rho\sim r$ and isolates the dependence on the matter content through $f(r)$.

\paragraph{Asymptotics, normalization, and global charges.}
The constant $\rho_0$ in \eqref{expKiselev} fixes the overall scale. In asymptotically flat cases one can choose $\rho_0$ so that $\rho/r\to 1$ as $r\to\infty$. Then $\Phi(\rho)\to 1$, and the asymptotic expansion of $\Phi$ yields the ADM mass \cite{arnowitt2008dynamics}:
\begin{equation}
\Phi(\rho)=1+\frac{M_{\mathrm{ADM}}}{2\rho}+O(\rho^{-2}).
\end{equation}
Equivalently, $M_{\mathrm{ADM}}$ is set by the coefficient of $1/\rho$ in the expansion of $r(\rho)/\rho$.

For non-asymptotically-flat cases (e.g.\ K\"ottler-like static patches), the same scaling can be fixed relative to a reference point or to the cosmological-horizon endpoint $\rho_\uparrow$. In that setting $\rho_0$ is best viewed as a convenient scale for the finite interval $(\rho_+,\rho_\uparrow)$.

\paragraph{Horizon behavior and coordinate regularity.}
In curvature coordinates the spatial metric diverges at a nonextremal horizon because $f^{-1}\to\infty$ as $r\to r_+^+$. In isotropic coordinates the spatial metric is $\Phi^4\delta_{ij}$, and the horizon lies at a finite radius $\rho_+$. By Theorem~\ref{thm:main}, $\rho(r)$ is strictly monotone on the exterior static region. The physical branch is therefore single-valued.

The local behavior depends on the horizon order. For a simple root, $\rho(r)\to\rho_+$ with a square-root law (Lemma~\ref{lem:SimpleHorizon}), and
$r(\rho)-r_+\propto(\rho-\rho_+)^2$.
Thus $r(\rho)$ has a stationary point at the horizon, but the inverse remains regular on the exterior branch. For extremal horizons, logarithmic terms can appear in $\ln\rho$ (Remark~\ref{rem:Extremal}). This can reduce analyticity of $r(\rho)$ and degrade convergence of local inversion series. The integral defining the map still exists, but the inverse is less regular.

%===============================================================
%===============================================================

\subsection{Parameter dependence, applications, and implementation}

For the multi-source family \eqref{f(r)}, the map depends on the matter content only through $f(r)$ and hence through $(r_g,K_i,w_i)$. Two points are useful.

\emph{(i) Separation of radial scalings.}
Each source contributes a term $r^{-(3w_i+1)}$ to $f(r)$. The series in Eq.~\eqref{eq:I(r)} combines these contributions through multinomial products. On a chosen domain one can treat $I(r)$ as an expansion in $\xi(r)=1-f(r)$. This isolates how each component modifies the geometry relative to Schwarzschild/RN.

\emph{(ii) Endpoints of the static region.}
If $f(r)>0$ for all large $r$ (asymptotically flat/AdS cases), then $\rho\in(\rho_+,\infty)$ and the Euclidean normalization at infinity is available. If there is an outer horizon (e.g.\ de Sitter-like asymptotics), then $\rho$ lies in a finite interval. In that case the nearest complex singularity of $r(\rho)$ need not coincide with the real endpoints. This motivates the convergence diagnostics in \S\ref{sec:InvConvergence}, which identify where a local inversion series is reliable.

Many calculations use a background-plus-perturbations split: field scattering, quasinormal modes, lensing/shadows, and waveform corrections from environmental structure \cite{Cardoso2019nature,Barack2019roadmap,Cole2023environments,Speri2023accretion}. In curvature coordinates the factor $f^{-1}(r)\mathrm{d}r^2$ is singular at horizons and can complicate near-horizon analysis.

Isotropic coordinates give two practical benefits.

\emph{(i) Conformally flat spatial data.}
The spatial metric is conformally flat, which aligns with standard conformal methods for specifying initial data and checking constraints. This is useful for perturbations on non-vacuum backgrounds and for numerical implementations built on conformal decompositions.

\emph{(ii) Compatibility with conformal numerical methods.}
Many approaches assume or exploit conformally flat spatial slices. The explicit map $\rho(r)$ and the inversion method provide a direct route to represent multi-source backgrounds in that framework.

For practical implementation, the computational workflow described in Sec.~\ref{sec:error_analysis_double_truncation} is summarized step by step in Appendix~\ref{app:pseudocode}, including prerequisites, patchwise local inversion, and refinement checks.

%===============================================================
%===============================================================

\subsection{Limitations and scope}

We assume static spherical symmetry and the restricted form \eqref{ds} with $g_{tt}=-f$ and $g_{rr}=f^{-1}$ on the exterior static region. This includes many phenomenological models but excludes the general static spherical line element with independent functions $f(r)$ and $h(r)$, i.e.\ $g_{rr}=h(r)^{-1}$. The defining ODE extends directly, with $(r\sqrt{h(r)})^{-1}$ replacing $(r\sqrt{f(r)})^{-1}$. The series organization would then depend on the chosen functional form of $h$.

Rotation and time dependence require different constructions. For stationary axisymmetric metrics one typically uses quasi-isotropic coordinates. For time-dependent cosmological black holes, isotropic-type coordinates can be defined only after specifying a different ansatz. The present results provide the static spherical case needed in such settings.

%############################################################
%#################   Conclusion   ##########################
%###########################################################

\section{Conclusion}\label{sec:conclusion}

The main result of this work is a constructive isotropic-coordinate map for the class \eqref{ds} with generalized Kiselev-type metric function \eqref{f(r)}, together with a practical route for recovering $r(\rho)$ on the exterior static region. Many analytic and numerical tasks simplify when the $t=\mathrm{const}$ spatial metric is conformally flat, whereas curvature (areal-radius) coordinates obscure this structure and introduce a coordinate divergence in the spatial line element at nonextremal horizons through $f(r)^{-1}$.

Our construction, given in Theorem~\ref{thm:main}, provides an isotropic radius $\rho$ for the class \eqref{ds} such that
\begin{equation*}
\gamma_{ij}\,\mathrm{d}x^i\mathrm{d}x^j
=
\Phi(\rho)^4(\mathrm{d}\rho^2+\rho^2\mathrm{d}\Omega^2),
\qquad
\Phi(\rho)^2=\frac{r(\rho)}{\rho}.
\end{equation*}
Up to an overall scale, the map is fixed by the defining first-order relation
\begin{equation*}
\frac{1}{\rho}\frac{\mathrm{d}\rho}{\mathrm{d}r}
=
\frac{1}{r\sqrt{f(r)}}
\end{equation*}
and its integrated form \eqref{expKiselev}, so that the dependence on the matter content is encoded through $f(r)$ alone, including the multi-source family \eqref{f(r)}.

On the physical exterior branch, the map is well defined and strictly monotone by Theorem~\ref{thm:main}. Nonextremal horizons are sent to a finite isotropic radius, so the spatial metric no longer exhibits the curvature-coordinate blow-up. The distinction between simple and extremal horizons is also important in practice, since it governs the local regularity of the inverse map and the behavior of local series methods.

A second outcome of the paper is constructive access to $r(\rho)$, not only to $\rho(r)$. The inversion framework developed in Sec.~\ref{SecLagInv}, together with the truncation and validation strategy of Sec.~\ref{sec:error_analysis_double_truncation}, provides workable local representations even when no closed form is available. This gives controlled local access to the inverse map for geometric diagnostics, perturbative calculations, and numerical implementations, with additional diagnostics and implementation details, including pseudocode, collected in the appendices.

Overall, the isotropic map provides a common framework for the representation and analysis of Schwarzschild-like backgrounds with multiple sources. The conformal factor $\Phi(\rho)$ determines the spatial geometry. The nontrivial part of the coordinate change is encoded in the deviation from $\rho\sim r$, organized by $I(r)$. Within the class \eqref{ds}, horizons occur at finite isotropic radius, the dependence on matter parameters enters through $f(r)$, and the inverse map can be computed systematically for analytic and numerical calculations.

%############################################################
%#################   Data Availability  ####################
%###########################################################

\section*{Data Availability Statement}
Data sharing is not applicable to this article as no new data were created or analyzed in this study.

%############################################################
%###################     Appendix    #######################
%###########################################################

\appendix

%===============================================================
%===============================================================

\section{Notation}
\label{app:notation}

For convenience, we collect here the principal notations used throughout the paper.
\begin{center}
\renewcommand{\arraystretch}{1.2}
\begin{longtable}{p{0.16\textwidth} p{0.46\textwidth} p{0.30\textwidth}}
\hline
\textbf{Symbol} & \textbf{Meaning} & \textbf{Comments} \\
\hline
\endfirsthead

\hline
\textbf{Symbol} & \textbf{Meaning} & \textbf{Comments} \\
\hline
\endhead

\hline
\endfoot

$t$ & Static time coordinate & Unchanged under the curvature $\leftrightarrow$ isotropic radial transformation. \\

$r$ & Areal (curvature) radius & Defined by the area of symmetry 2-spheres, $4\pi r^2$; $r>0$. \\

$\rho$ & Isotropic radius & Defined implicitly by $\rho(r)=\rho_0 r e^{I(r)}$ on the static domain. \\

$\theta,\phi$ & Angular coordinates on the unit 2-sphere & Standard spherical coordinates. \\

$\mathrm{d} s^2$ & Spacetime line element & Metric interval. \\

$\mathrm{d} \Omega^2$  & Metric on the unit 2-sphere & $\mathrm{d} \Omega^2=\mathrm{d}\theta^2+\sin^2\theta\,\mathrm{d}\phi^2$. \\

$g_{\mu\nu}$ & Spacetime metric & Mostly-plus signature $(-,+,+,+)$. \\

$R^{\mu}{}_{\nu}$ & Ricci tensor & Appears in the Einstein equations. \\

$R$ & Ricci scalar & Scalar curvature. \\

$S$ & Action & Einstein-Hilbert plus matter action. \\

$T^{\mu}{}_{\nu}$ & Total stress--energy tensor &  Here, $T^{\mu}{}_{\nu}=\mathrm{diag}(-\varepsilon,p_r,p_\perp,p_\perp)$. \\

$T^{\mu}{}_{\nu(i)}$ & Stress--energy tensor of the $i$th component & Individual noninteracting anisotropic-fluid contribution. \\

$\varepsilon$  & Energy density & $\varepsilon>0$. \\

$p_r$ & Radial pressure & Principal pressure in the radial direction. \\

$p_\perp$ & Tangential pressure & Degenerate angular principal pressure. \\

$w_i$ & Effective equation-of-state parameter of the $i$th component & Admissible range depends on the convention in Sec.~\ref{sec:setup}; $w_i=0$ is degenerate with the Schwarzschild sector. \\

$N$ & Number of nonvacuum matter components & Assumed finite. \\

$f(r)$ & Metric function in curvature coordinates &  \\

$F(\rho)$ & Squared lapse function in isotropic coordinates & $g_{tt}=-F(\rho).$ \\

$\Phi(\rho)$ & Conformal factor of the spatial metric in isotropic coordinates & \\

$M$ & Black hole mass parameter &  \\

$r_g$ & Gravitational radius & $r_g=2M$. \\

$K_i$ & Amplitude of the $i$th Kiselev-type matter term & May have either sign in the general ansatz; controls the strength of the corresponding power-law contribution. \\

$\Lambda$ & Cosmological constant & Used in the K\"ottler example. \\

$Q$ & Electric charge parameter & Used in the RN and RN$+$Kiselev examples. \\

$r_+$ & Event horizon & Usually the lower endpoint of the static exterior domain when present. \\

$r_\uparrow$ & Upper endpoint of the static domain & Equals $+\infty$ for asymptotically flat/AdS cases, or the cosmological horizon radius when a bounded static patch exists. \\

$\mathcal{S}$ & Exterior static domain & Typically $\mathcal{S}=(r_+,r_\uparrow)$. \\

$\overline{\mathcal{S}}$ & Closure of the static radial domain $\mathcal{S}$ & Includes the boundary points of $\mathcal{S}$, such as $r_+$ or $r_{\uparrow}$ when they exist. \\

$r_\ast$ & Base point for the local Lagrange-B\"urmann inverse & Must satisfy $f(r_\ast)>0$; used in Sec.~\ref{SecLagInv} and Algorithm~\ref{alg:inverse-lagrange-consistent}. \\

$\rho_0$ & Multiplicative normalization constant in the isotropic map & Fixes the overall scale of $\rho$; chosen conveniently in special cases. \\

$\rho_+$ & Isotropic radius corresponding to $r_+$ & Finite in standard nonextremal examples. \\

$\rho_\uparrow$ & Upper endpoint of the isotropic radial interval & Image of $r_\uparrow$ under the map $\rho(r)$. \\

$\Psi$ & Forward radial map used in the inversion section & Usually $\Psi(r)=\rho(r)$. \\

$g(r)$ & Auxiliary function in the Lagrange inversion formula & $g(r)=\Psi'(r)/\Psi(r)=1/\bigl(r\sqrt{f(r)}\bigr)$. \\

$I(r)$ & Logarithmic correction in the isotropic map & Defined by $I(r)=\int \frac{\mathrm{d} x}{x}\bigl(f(x)^{-1/2}-1\bigr)$, so that $\rho(r)=\rho_0 r e^{I(r)}$. \\

$\xi(r)$ & Deviation from unity in the metric function & Defined by $f(r)=1-\xi(r)$ in the series expansion approach. \\

$\xi_n(r)$ & Individual term in the decomposition of $\xi(r)$ & For example $\xi_0(r)=r_g/r$ and $\xi_n(r)=K_n/r^{s_n}$. \\

$s_n$ & Power exponent of the $n$th source term in $\xi_n(r)=K_n/r^{s_n}$ & In Kiselev-type terms  $s_n=3w_n+1$. \\

$\mathbf m=(m_0,m_1,\dots,m_N)$ & Multi-index used in the multinomial expansion & All $m_i\ge 0$ and $|\mathbf m|=\sum_i m_i$. \\

$|\mathbf m|$ & Order of the multi-index & Equals the total multinomial order $k$. \\

$m_0,m_1,\dots,m_N$ & Components of the multi-index $\mathbf m$ & Count how often each term appears in the expansion. \\

$C(\mathbf m)$ & Combinatorial coefficient in the multi-index expansion & Built from binomial/multinomial factors. \\

$\alpha(\mathbf m)$ & Total radial exponent associated with $\mathbf m$ & Governs the power $x^{-1-\alpha(\mathbf m)}$ appearing under the integral. \\

$c_j$ & Coefficients in the generalized power sum $f(r)=\sum_{j=1}^m c_j r^{\alpha_j}$ & Used in the root-counting argument. \\

$\alpha_j$ & Distinct real exponents in the generalized power sum & After merging coincident powers; determine the Chebyshev-system structure. \\

$m$ & Number of distinct exponents in the generalized power sum & Gives the upper bound $m-1$ on the number of positive roots. \\

$h(x)$ & Exponential sum obtained after the change of variable $x=\ln r$ & Defined by $h(x)=f(e^x)$. \\

$x$ & Auxiliary radial variable & Used either as an integration dummy variable or as $x=\ln r$ in the root-counting proof. \\

$W(x)$ & Wronskian of the exponential basis $\{e^{\alpha_j x}\}$ & Nonzero for distinct $\alpha_j$, proving the Chebyshev property. \\

$r_1,r_2,r_3$ & The three real roots in the K\"ottler factorization & Ordered in the draft as $r_1>r_2>0>r_3$ for the de Sitter case. \\

$\phi(r)$ & Amplitude used in the elliptic-integral representation for the K\"ottler map & Defined from the radial variable in the K\"ottler example. \\

$k$ & Elliptic modulus / series order (context dependent) & In the K\"ottler example, $k^2=\frac{-\,r_3(r_1-r_2)}{\,r_1(r_2-r_3)\,}\in(0,1)$; elsewhere $k$ is also the summation order in series expansions. \\

$F(\phi,k)$ & Incomplete elliptic integral of the first kind & Used for the explicit K\"ottler isotropic map. \\

$u(\rho)$ & Auxiliary elliptic-function variable in the K\"ottler inversion & Defined from $\ln(\rho/\rho_+)$. \\

$\mathrm{sn}(u,k)$ & Jacobi elliptic sine & Appears in the inversion of the K\"ottler map. \\

$\Psi'(r_\ast), \Psi''(r_\ast),\Psi'''(r_\ast)$ & Derivatives of the forward map at the expansion point & Enter the low-order Lagrange-{B\"u}rmann inverse coefficients. \\

$c_1,c_2,c_3$ & First coefficients of the local inverse series $r(\rho)$ & Given explicitly in the inversion section. \\

$\mathcal{A}$ & Auxiliary quantity introduced in the inversion formulas & Defined in the Lagrange-inversion section. \\

$R$ & Radius of convergence of the inverse series & Determined by the nearest obstruction to analytic continuation. \\

$c_n$ & Generic coefficients of the inverse series & Used in the convergence diagnostics. \\

$\varepsilon_{\mathrm I}^{(k_{\max})}$ & Forward-map truncation error & Error induced by truncating the multi-index expansion of $I(r)$ at total order $|\mathbf m|\le k_{\max}$. \\

$\varepsilon_{\mathrm L}^{(N)}$ & Inverse-series truncation error & Error induced by truncating the Lagrange inverse series for $r(\rho)$ at order $N$. \\

$\mathbb{R},\mathbb{N},\mathbb{Q},\mathbb{Z}$ & Standard number sets & Real, natural, rational, and integer numbers, respectively. \\

\hline
\end{longtable}
\end{center}

%===============================================================
%===============================================================

\section{A sharp bound on the number of positive roots of the generalized Kiselev-type metric function}
\label{app:rootbound}

Consider the metric function \eqref{f(r)} on $r>0$.
Although the exponents $3w_i+1$ may be fractional (or even irrational), the number of \emph{positive} real roots of $f$ admits a simple universal bound that depends only on the number of \emph{distinct} power-law terms in $f$.

\begin{prop}[Maximum number of positive roots]
\label{prop:max_positive_roots}
Let $m$ be the number of distinct exponents among $\{0,-1,-(3w_i+1)\}_{i=1}^{N}$ after merging coincident values (and deleting zero coefficients). Then the equation $f(r)=0$ has at most $m-1$ solutions in $r>0$, counted with multiplicity.
\end{prop}

\begin{proof}
\medskip
\noindent\emph{Reduction to an exponential sum.}

Write \eqref{f(r)} as a finite linear combination of real powers on $(0,\infty)$:
\begin{equation}
f(r)=\sum_{j=1}^{m} c_j\, r^{\alpha_j},
\label{eq:gen_poly_form}
\end{equation}
where the exponents $\alpha_j\in\mathbb{R}$ are \emph{distinct} (combine any coincident exponents into a single term, and discard any term with zero coefficient). In our case, the candidate exponent set is
\begin{equation}
\{0,\,-1,\,-(3w_1+1),\ldots,-(3w_N+1)\},
\end{equation}
and $m$ denotes the number of distinct exponents actually present after merging coincident values.

Introduce the logarithmic coordinate $x=\ln r$ (a bijection between $r\in(0,\infty)$ and $x\in\mathbb{R}$), and define
\begin{equation}
h(x):=f(e^x)=\sum_{j=1}^{m} c_j\, e^{\alpha_j x}.
\label{eq:exp_sum}
\end{equation}
Then $f(r)=0$ for some $r>0$ if and only if $h(x)=0$ for $x=\ln r$, and the multiplicities are preserved because $r\mapsto \ln r$ is smooth and strictly monotone.

\medskip
\noindent\emph{Chebyshev-system property and the $m-1$ zero bound.}

The family $\{e^{\alpha_1 x},\ldots,e^{\alpha_m x}\}$ with distinct real $\alpha_j$ is a classical \emph{(extended) Chebyshev system} on any interval (indeed, on all of $\mathbb{R}$) \citep{karlinStudden1966,pinkus2010}. One convenient certificate is the Wronskian:
\begin{align}
W(x)
=\det\!\Bigl(\frac{d^{k-1}}{dx^{k-1}}e^{\alpha_j x}\Bigr)_{k,j=1}^{m}
&=\det\!\bigl(\alpha_j^{k-1}e^{\alpha_j x}\bigr)_{k,j=1}^{m}\nonumber\\
=e^{(\alpha_1+\cdots+\alpha_m)x}\,
\det\!\bigl(\alpha_j^{k-1}\bigr)_{k,j=1}^{m}
&=e^{(\alpha_1+\cdots+\alpha_m)x}\,\prod_{1\le p<q\le m}(\alpha_q-\alpha_p),
\label{eq:wronskian_exp}
\end{align}
which is nonzero for all $x\in\mathbb{R}$ when the $\alpha_j$ are distinct. A standard consequence of the Chebyshev-system property is that any nontrivial linear combination \eqref{eq:exp_sum} has at most $m-1$ real zeros on $\mathbb{R}$, counting multiplicity \citep{karlinStudden1966,pinkus2010}. Therefore the original function $f$ has at most $m-1$ positive zeros on $(0,\infty)$.
\end{proof}

\paragraph{Generic Kiselev sums.}
In the generic situation where all $w_i$ are distinct, all $K_i\neq 0$, and no Kiselev exponent coincides with the $r^{-1}$ term (i.e.\ no $w_i=0$), we have $m=N+2$, hence
\begin{equation}
\#\{r>0:\ f(r)=0\}\ \le\ N+1
\qquad\text{(counting multiplicity)}.
\label{eq:generic_bound}
\end{equation}
If some $w_i$ coincide (or if some $w_i=0$ so that $r^{-(3w_i+1)}=r^{-1}$ merges with the $1/r$ term), then $m$ decreases accordingly and the bound strengthens.

\begin{remark}
This bound is purely structural: it does \emph{not} require any sign restriction on $K_i$, nor integrality/rationality of exponents. It is also essentially sharp in the sense that exponential sums with $m$ terms can realize $m-1$ simple real zeros for suitable coefficients \citep{novikovShapiro2018}.
\end{remark}

%===============================================================
%===============================================================

\section{K\"ottler Metric in Isotropic Coordinates}\label{app:Kottler-isotropic}

In this appendix we derive the K\"ottler metric in isotropic coordinates used in Sec.~\ref{sec:Kottler_example}. We use the incomplete elliptic integral of the first kind and the Jacobi elliptic sine,
\begin{align}
    &u = F(\phi, k) := \int_0^{\phi}\frac{\mathrm{d}\theta}{\sqrt{\,1-k^2\sin^2\theta\,}},\\
    &\mathrm{sn}(u, k) := \sin\phi.
\end{align}
The K\"ottler solution of Einstein's equations is
\begin{equation}
    \mathrm{d} s^2 = -\Bigl( 1 - \frac{r_g}{r} - \frac{\Lambda}{3}r^2 \Bigr)\,\mathrm{d} t^2
    + \Bigl( 1 - \frac{r_g}{r} - \frac{\Lambda}{3}r^2 \Bigr)^{-1} \mathrm{d} r^2
    + r^2\mathrm{d}\Omega^2.
\end{equation}
Passing to isotropic radius \(\rho\) via
\begin{equation}
    \Bigl( 1 - \frac{r_g}{r} - \frac{\Lambda}{3}r^2 \Bigr)^{-1}\mathrm{d} r^2
    = \Bigl(\frac{r}{\rho}\Bigr)^2 \mathrm{d} \rho^2,
    \qquad
    r^2 \mathrm{d}\Omega^2 = \Bigl(\frac{r}{\rho}\Bigr)^2 \rho^2 \mathrm{d}\Omega^2,
\end{equation}
gives the differential equation
\begin{equation}
    \frac{\mathrm{d} \rho}{\rho}
    =\frac{\mathrm{d} r}{r\sqrt{1-\frac{r_g}{r}-\frac{\Lambda}{3}r^2}}.
\end{equation}
Factor the cubic under the root:
\begin{equation}
    1-\frac{r_g}{r}-\frac{\Lambda}{3}r^2
    = -\frac{\Lambda}{3r}\,(r-r_1)(r-r_2)(r-r_3),
\end{equation}
where \(r_i\) are the three real roots of
\begin{equation}
    -\frac{\Lambda}{3}r^3 + r - r_g = 0,
    \qquad
    r_1>r_2>0>r_3,
    \qquad
    r_1+r_2+r_3=0.
\end{equation}
Recall that we identify the horizons as
\(r_{\uparrow}:=r_1,\; r_+:=r_2\),
so that for \(\Lambda>0\) the static region is
\(\mathcal{S}=(r_+,r_{\uparrow})=(r_2,r_1)\).
Choosing \(\rho(r_2)=\rho_+\), we obtain
\begin{equation}
    \ln\frac{\rho}{\rho_+}
    =\sqrt{\frac{3}{\Lambda}}
    \int_{r_2}^{r}\frac{\mathrm{d} x}{\sqrt{-\,x(x-r_1)(x-r_2)(x-r_3)}}.
\end{equation}
Introduce the Legendre normal form with
\begin{equation}
    \phi(r)=\arcsin\sqrt{\frac{\,r_1(r-r_2)\,}{\,r(r_1-r_2)\,}},
    \qquad
    k^2=\frac{-\,r_3(r_1-r_2)}{\,r_1(r_2-r_3)\,}\in(0,1).
\end{equation}
Then the integral evaluates to the incomplete elliptic integral of the first kind:
\begin{equation}
    \ln\frac{\rho}{\rho_+}
    =\sqrt{\frac{3}{\Lambda}}\;
    \frac{2}{\sqrt{\,r_1(r_2-r_3)\,}}\;
    F\!\bigl(\phi(r), k\bigr).
\end{equation}
Equivalently,
\begin{equation}
    \rho(r)=\rho_+\,
    \exp\!\left[
    \sqrt{\frac{3}{\Lambda}}\;
    \frac{2}{\sqrt{\,r_1(r_2-r_3)\,}}\;
    F\!\left(
    \arcsin\sqrt{\frac{\,r_1(r-r_2)\,}{\,r(r_1-r_2)\,}},
    k
    \right)\right].
\end{equation}
For the physical static region \(r\in(r_2,r_1)=\mathcal{S}\) (between the event and cosmological horizons when \(\Lambda>0\)), the radicand in the quartic root is positive, so the transformation is real on the interval relevant for Theorem~\ref{thm:main}.

To invert the coordinate transformation, define
\begin{equation}
    u(\rho)=\frac{\sqrt{\,r_1(r_2-r_3)\,}}{2}\,\sqrt{\frac{\Lambda}{3}}\,\ln\frac{\rho}{\rho_+}.
\end{equation}
Then \(\sin\phi=\mathrm{sn}(u(\rho),k)\), and hence
\begin{equation}
\frac{r_1(r-r_2)}{r(r_1-r_2)}
=\mathrm{sn}^2\!\bigl(u(\rho),k\bigr).
\end{equation}
Solving for \(r\) yields the Jacobi form
\begin{equation}\label{eq:Jacobi}
    r(\rho)=\frac{\,r_1 r_2\,}{\;r_1-(r_1-r_2)\,\mathrm{sn}^2\!\bigl(u(\rho), k\bigr)\;}.
\end{equation}
Substituting \(r(\rho)\) back into the line element gives the K\"ottler metric in isotropic (conformally flat) form
\begin{equation}
    \mathrm{d} s^2_{\mathrm{Kottler}}
    = -\Bigl[ 1-\frac{r_g}{r(\rho)}-\frac{\Lambda}{3}r(\rho)^2 \Bigr]\,\mathrm{d} t^2
    \;+\;
    \Bigl(\frac{r(\rho)}{\rho}\Bigr)^{\!2}\,\bigl(\mathrm{d} \rho^2+\rho^2\mathrm{d}\Omega^2\bigr).
\end{equation}

%===============================================================
%===============================================================

\section{Evaluating derivatives of $g$ from the multi-index expansion}\label{app:LagrangeAux}

In any region where the expansion of $f(r)^{-1/2}=(1-\xi(r))^{-1/2}$ used in the proof of Proposition~\ref{prop:multiindex} (see Eq.~\eqref{expansion1}) is valid
(in particular, where the chosen branch of $f^{-1/2}$ is analytic and $|\xi(r)|<1$), the quantity
$g(r)=\bigl(r\sqrt{f(r)}\bigr)^{-1}$ admits the corresponding multi-index representation.

Let $\mathbf m=(m_0,m_1,\dots,m_N)\in\mathbb{N}_0^{\,N+1}$ be a multi-index with
$|\mathbf m|:=\sum_{i=0}^N m_i$. Using the notation of Proposition~\ref{prop:multiindex}, namely
\begin{equation}
  s_i:=3w_i+1,
\end{equation}
together with the associated quantities $\alpha(\mathbf m)$ and $C(\mathbf m)$, we obtain
\begin{equation}\label{eq:gMulti}
  g(r)=\frac{1}{r\sqrt{f(r)}}
  =\sum_{\mathbf m}C(\mathbf m)\,r_g^{m_0}\Bigl(\prod_{i=1}^{N}K_i^{m_i}\Bigr)\,r^{-1-\alpha(\mathbf m)},
\end{equation}
where the sum includes the term $|\mathbf m|=0$ (which reproduces the leading $1/r$ behavior).
Differentiating the series term by term gives, for $j\in\mathbb{N}$,
\begin{equation}\label{eq:gDerivsMulti}
  g^{(j)}(r)=(-1)^j\sum_{\mathbf m}C(\mathbf m)\,r_g^{m_0}\Bigl(\prod_{i=1}^{N}K_i^{m_i}\Bigr)\,
  \left(\prod_{\ell=1}^{j}\bigl(\alpha(\mathbf m)+\ell\bigr)\right)\,r^{-1-\alpha(\mathbf m)-j}.
\end{equation}
Equations \eqref{eq:rLowOrder}--\eqref{eq:gDerivsMulti} therefore provide explicit Taylor approximations to $r(\rho)$ about any base point $r_\ast$ lying in a region where the expansion in Eq.~\eqref{expansion1} is valid. In practice, for many standard parameter regimes in physical applications, the expansion in Eq.~\eqref{expansion1} is valid on open neighborhoods of regular points in the exterior static region; the formulas are therefore directly usable at such points.

%===============================================================
%===============================================================

\section{Analytic continuation and truncation of the inverse series}\label{app:InvConvergence}

\subsection{Analytic radius and obstructions to continuation}

Theorem~\ref{LagrangeInversion} guarantees that the Lagrange-B{\"u}rmann series
\eqref{eq:KiselevInv} has a nonzero radius of convergence at any base point $r_\ast$ with
$\Psi'(r_\ast)\neq 0$. We now (i) describe what controls the convergence disk in the
$\rho$-plane, and (ii) summarize practical diagnostics for the convergence radius and asymptotic
rate.

Let $r(\rho)$ denote the analytic continuation of the local inverse of $\Psi(r)=\rho(r)$
constructed near $\rho_\ast=\Psi(r_\ast)$. For a holomorphic function, the radius of
convergence of the Taylor expansion of $r(\rho)$ about $\rho_\ast$ equals the distance from
$\rho_\ast$ to the nearest point at which $r(\rho)$ cannot be continued as a single-valued
holomorphic function.

In complex analysis, the radius of convergence of the Taylor series of $r(\rho)$ about $\rho_\ast$ is the distance from $\rho_\ast$ to the nearest point at which $r(\rho)$ cannot be continued as a single-valued holomorphic function. Defining the singular set
\begin{equation}
\Sigma_{\rho}:=\Bigl\{\rho\in\mathbb{C}:\ r(\rho)\ \text{cannot be continued as a single-valued holomorphic function in a neighborhood of }\rho\Bigr\},
\end{equation}
one therefore has
\begin{equation}
|\rho-\rho_\ast|<R,\qquad
R=\operatorname{dist}(\rho_\ast,\Sigma_{\rho})
=\inf_{\rho_s\in\Sigma_{\rho}}|\rho_s-\rho_\ast|\,.
\label{eq:Rdef}
\end{equation}

On the real static interval $\mathcal{S}$, Eq.~\eqref{eq:g-def}, Eq.~\eqref{eq:PsiODE}, and the positivity of both $f(r)$ and $\rho(r)$ imply that $\Psi'(r)\neq 0$. Thus the real inverse has no critical-point obstruction on $\rho(\mathcal{S})$.

Under complex continuation, obstructions to analytic continuation of the inverse can arise from singularities or branch points of $\Psi(r)$, induced by those of $g(r)=1/(r\sqrt{f(r)})$, namely from $r=0$, from zeros or branch points of $f(r)$ (including complex roots and branch cuts when $s_i=3w_i+1\notin\mathbb{Z}$). The nearest such obstruction in the $\rho$-plane controls $R$ in \eqref{eq:Rdef}; importantly, the dominant obstruction need not lie on the real axis.

By Theorem~\ref{thm:main}, $\Psi=\rho$ is strictly increasing on $\mathcal{S}$ and thus admits a unique real inverse on $\rho(\mathcal{S})$. Nevertheless, the one-sided endpoint values of $\rho$ associated with finite endpoints of $\mathcal{S}$ need not coincide with the nearest singularity of the analytic continuation of $r(\rho)$. The next lemma shows that a nonextremal horizon produces $\mathrm{d}r/\mathrm{d}\rho=0$ at $\rho=\rho_h$, but does not by itself obstruct analyticity of the inverse.

%===============================================================

\subsection{Coefficient diagnostics: Cauchy--Hadamard and Domb--Sykes}\label{app:CauchyHadamard}

Let
\begin{equation}
r(\rho)=r_\ast+\sum_{n=1}^{\infty}c_n(\rho-\rho_\ast)^n
\label{eq:rSeries}
\end{equation}
denote the inverse series about $\rho_\ast$. The radius $R$ in \eqref{eq:Rdef} can be estimated
directly from the coefficients via the Cauchy--Hadamard formula~\cite{Hardy1916-CauchyHadamard},
\begin{equation}
\frac{1}{R}=\limsup_{n\to\infty}|c_n|^{1/n}\,.
\label{eq:CauchyHadamard}
\end{equation}
In practice, for sufficiently large $n$, one may use the ratio and root approximants
\begin{equation}
R^{(\mathrm{ratio})}_n=\left|\frac{c_n}{c_{n+1}}\right|,
\qquad
R^{(\mathrm{root})}_n=|c_n|^{-1/n},
\label{eq:RratioRoot}
\end{equation}
and interpret their stabilization (or lack thereof) as an indicator of the dominant obstruction to analytic continuation.

More detailed information about the nature of the nearest obstruction to analytic continuation,
when that obstruction is algebraic, can be extracted from a Domb--Sykes
analysis~\cite{DombSykes1956,domb1957susceptibility,hunter1980deducing}. Assume that the
closest obstruction occurs at $\rho=\rho_s$ with $|\rho_s-\rho_\ast|=R$ and that near $\rho_s$
the inverse has an algebraic singular expansion of the form
\begin{equation}
r(\rho)=r_{\mathrm{reg}}(\rho)+\mathcal{K}\!\left(1-\frac{\rho-\rho_\ast}{\rho_s-\rho_\ast}\right)^{\lambda},
\qquad \lambda\notin\mathbb{Z}_{\ge 0},
\label{eq:AlgSingAnsatz}
\end{equation}
where $r_{\mathrm{reg}}$ is analytic at $\rho_s$ and $\mathcal{K}$ is a constant. Then the coefficients have the asymptotic form
\begin{equation}
c_n \sim C\,(\rho_s-\rho_\ast)^{-n}\,n^{-\lambda-1}\qquad (n\to\infty),
\label{eq:cnAsymptotic}
\end{equation}
and the coefficient ratios satisfy
\begin{equation}
\frac{c_n}{c_{n-1}}
=
\frac{1}{\rho_s-\rho_\ast}\left(1-\frac{\lambda+1}{n}+O(n^{-2})\right).
\label{eq:DombSykes}
\end{equation}
Thus, when the dominant obstruction is a single real algebraic singularity, a linear fit of
$c_n/c_{n-1}$ versus $1/n$ for large $n$ yields estimates for $R=|\rho_s-\rho_\ast|$ and
$\lambda$. If the closest obstructions occur as a complex conjugate pair
$\rho_s-\rho_\ast=Re^{\pm i\theta}$, then $c_n$ typically exhibits oscillations
$c_n\sim R^{-n}n^{-\lambda-1}\cos(n\theta+\phi)$, in which case the root test
\eqref{eq:RratioRoot} and fits to $|c_n|$ are more robust.

%===============================================================

\subsection{Truncation error and convergence speed}

Let $S_N(\rho)=r_\ast+\sum_{n=1}^{N}c_n(\rho-\rho_\ast)^n$ denote the $N$-th truncation of
\eqref{eq:rSeries}, and define the remainder $E_N(\rho)=r(\rho)-S_N(\rho)$. For any $\widetilde R$ with $0<\widetilde R<R$, the function $r(\rho)$ is holomorphic in the closed disk $|\rho-\rho_\ast|\le \widetilde R$. Define
\begin{equation}
M_{\widetilde R}:=\max_{|\rho-\rho_\ast|=\widetilde R}|r(\rho)|.
\end{equation}
Then Cauchy estimates~\cite{ahlfors1979complex} imply the coefficient bound $|c_n|\le M_{\widetilde R}/\widetilde R^n$ and the uniform remainder estimate
\begin{equation}
|E_N(\rho)|
\le
\frac{M_{\widetilde R}}{\widetilde R-|\rho-\rho_\ast|}\left(\frac{|\rho-\rho_\ast|}{\widetilde R}\right)^{N+1},
\qquad |\rho-\rho_\ast|<\widetilde R.
\label{eq:RemainderBound}
\end{equation}
Hence the convergence is \emph{geometric} in $N$ for any fixed $q:=|\rho-\rho_\ast|/\widetilde R<1$.
If, moreover, the dominant obstruction to analytic continuation is a single algebraic one as in \eqref{eq:AlgSingAnsatz}, then the large-$N$ decay is controlled by the true convergence radius $R$ from \eqref{eq:Rdef}, rather than by the auxiliary radius $\widetilde R$ used in the Cauchy bound. In that case \eqref{eq:cnAsymptotic} implies the refined asymptotic decay
\begin{equation}
|E_N(\rho)| \approx \mathrm{const}\times q^{\,N+1}(N+1)^{-\lambda-1},
\qquad q=\frac{|\rho-\rho_\ast|}{R},
\label{eq:ErrorAsymptotic}
\end{equation}
i.e.\ geometric convergence modulated by an algebraic factor determined by the singularity
exponent $\lambda$.

%===============================================================

\subsection{Practical validation of the truncated inverse}\label{app:series_diagnostics}

To verify the accuracy of the truncated transformation in practice, one may supplement the coefficient diagnostics of \S\ref{app:CauchyHadamard} with direct validation tests.

\paragraph{Richardson extrapolation.}
For validation, compute $r(\rho)$ at successively finer truncations $(k_{\max}, N) = (10, 15)$, $(20, 30)$, $(30, 45)$ and verify that the differences decay geometrically. Richardson extrapolation~\cite{burden2015numerical,baumgarte2010numerical,richardson1911approximate} can then estimate the true value and the convergence rate.

\paragraph{Comparison with exact cases.}
For Schwarzschild, RN, and K\"ottler metrics, explicit expressions for $r(\rho)$ are available; see the corresponding benchmark cases in \S\ref{sec:schwarzschild_example}, \S\ref{sec:rn_example}, and \S\ref{sec:Kottler_example}. Direct comparison of the truncated series with these benchmark formulas provides a stringent accuracy check and helps validate the error estimates in Table~\ref{tab:error_estimates}.

These tests provide an additional check that the truncated transformation is accurate enough for downstream numerical-relativity applications~\cite{loffler2012einstein,baumgarte2010numerical,alcubierre2008introduction}.
%===============================================================
%===============================================================

\section{Total Error Decomposition}\label{app:double_truncation_error}

The computation of a doubly truncated approximation $r^{(N,k_{\max})}(\rho)$ to $r(\rho)$ incurs three sources of error:
\begin{equation}
\left|r^{(N,k_{\max})}(\rho) - r(\rho)\right|
\leq
\varepsilon_{\text{I}}^{(k_{\max})} + \varepsilon_{\text{L}}^{(N)} + \varepsilon_{\text{coupling}}.
\label{eq:total_error_decomp}
\end{equation}
where\\

 $\varepsilon_{\text{I}}^{(k_{\max})}$ is the truncation error in the forward map $I(r)$,\\
 
 $\varepsilon_{\text{L}}^{(N)}$ is the truncation error in the Lagrange series (assuming exact $I(r)$),\\
 
 $\varepsilon_{\text{coupling}}$ accounts for interaction between the two approximations.

\paragraph{Forward map truncation error.}
Consider a fixed compact subinterval of $\mathcal{S}$ (or, more generally, a chosen region away from the endpoints) on which
\begin{equation}
\xi_{\max}:=\sup |\xi(r)|<1.
\end{equation}
On such a region, the binomial expansion of $(1-\xi(r))^{-1/2}$ converges absolutely and uniformly, so the corresponding multi-index representation of $I(r)$ may be truncated termwise. The resulting tail is
\begin{equation}
\varepsilon_{\text{I}}^{(k_{\max})}(r)
=
\left|
\sum_{|\mathbf m|>k_{\max}}
C(\mathbf m)\,r_g^{m_0}\prod_{i=1}^{N}K_i^{m_i}
\int_{r_+}^{r}\mathrm{d}x\,x^{-1-\alpha(\mathbf m)}
\right|.
\label{eq:forward_error_raw}
\end{equation}

\paragraph{Lagrange inverse truncation error.}
For the inverse series \eqref{eq:rSeries}, the truncation error estimates derived in
\S\ref{sec:InvConvergence} apply directly. In particular, the rigorous remainder bound is
given by Eq.~\eqref{eq:RemainderBound}, while, if the dominant obstruction to analytic continuation
is a single algebraic one as in \eqref{eq:AlgSingAnsatz}, the refined large-$N$ asymptotic estimate
is given by Eq.~\eqref{eq:ErrorAsymptotic}. Thus the inverse-series truncation error decays
geometrically in $N$, with an additional algebraic modulation in the algebraic-obstruction case.

\paragraph{Coupling error.}
The Lagrange coefficients $c_n$ depend on derivatives of $\rho(r)=\rho_0 r e^{I(r)}$. Replacing $I(r)$ by its truncation $I^{(k_{\max})}(r)$ therefore perturbs the coefficients as well. At the level of scaling on a chosen compact region with $|\xi(r)|\le \xi_{\max}<1$, one expects
\begin{equation}
c_n^{(k_{\max})}-c_n = O\!\left(\xi_{\max}^{k_{\max}+1}\right).
\end{equation}
Accordingly, for $|\rho-\rho_\ast|<R$, a schematic coupling estimate is
\begin{equation}
\varepsilon_{\text{coupling}}
\lesssim
C_3\,N\,\xi_{\max}^{k_{\max}+1}\,q^N,
\qquad
q:=\frac{|\rho-\rho_\ast|}{R}<1.
\label{eq:coupling_error}
\end{equation}
Thus, for fixed $k_{\max}$, increasing $N$ eventually ceases to improve the total error unless the forward truncation is refined as well.

\paragraph{Optimal parameter balance.}
To minimize the total error, one may balance the forward-tail scale controlled by $\xi_{\max}^{k_{\max}}$ against the inverse-series tail scale described by Eq.~\eqref{eq:ErrorAsymptotic}:
\begin{equation}
\xi_{\max}^{k_{\max}} \approx q^N \quad \Longrightarrow \quad N \approx k_{\max} \frac{\ln \xi_{\max}}{\ln q}.
\label{eq:optimal_balance}
\end{equation}
For typical values $\xi_{\max} \approx 0.3$ and $q \approx 0.5$, this yields $N \approx 1.7 \, k_{\max}$, suggesting that $N$ and $k_{\max}$ should be chosen in roughly equal proportion.

%===============================================================
%===============================================================

\section{Illustrative numerical estimates for specific examples}\label{app:error_estimates_examples}

We record schematic order-of-magnitude estimates for the Schwarzschild, RN, and K\"ottler solutions, intended only to illustrate the scaling relations discussed above.

\paragraph{Schwarzschild metric.}
For $f(r) = 1 - r_g/r$, we have $\xi(r) = r_g/r$. On the exterior region $r > r_+ = r_g$, typical evaluation points lie near $r \sim 3r_g$ (the photon sphere). At this radius,
\begin{equation}
\xi_{\max} \approx \frac{r_g}{3r_g} = 0.33.
\end{equation}
With the standard asymptotically flat normalization of the isotropic radius (see \S\ref{sec:schwarzschild_example}), the horizon is located at $\rho_+=r_g/4$, corresponding to $r_+=r_g$. In that normalization, the exact Schwarzschild inverse is
\begin{equation}
r(\rho)=\rho\left(1+\frac{r_g}{4\rho}\right)^2.
\end{equation}
the nearest singularity in the $\rho$-plane is at $\rho=0$. Thus, for the illustrative base point $\rho_\ast=r_g/2$, the convergence radius is
\begin{equation}
R=\frac{r_g}{2}.
\end{equation}
Taking, illustratively, $|\rho-\rho_\ast|\approx r_g/4$ then gives
\begin{equation}
q=\frac{|\rho-\rho_\ast|}{R}\approx 0.5.
\end{equation}
From \eqref{eq:optimal_balance}, setting $k_{\max} = 20$ suggests $N \approx 32$. The resulting errors are
\begin{align}
\varepsilon_{\text{I}}^{(20)} &\sim (0.33)^{20} \approx 2 \times 10^{-10}, \\
\varepsilon_{\text{L}}^{(32)} &\sim (0.5)^{32} \approx 2 \times 10^{-10},
\end{align}
yielding combined accuracy $\sim 10^{-9}$ to $10^{-10}$, sufficient for most numerical relativity applications.

\paragraph{Reissner--Nordstr\"om metric.}
For RN, with $f(r) = 1 - r_g/r + Q^2/r^2$,
\begin{equation*}
\xi(r)=\frac{r_g}{r}-\frac{Q^2}{r^2}.
\end{equation*}
At a representative radius $r=3r_g$, this gives
\begin{equation}
\xi(3r_g)=\frac{1}{3}-\frac{Q^2}{9r_g^2}.
\end{equation}
For example, if $Q^2=r_g^2/8$, then
\begin{equation}
\xi(3r_g)=\frac{1}{3}-\frac{1}{72}=\frac{23}{72}\approx 0.32.
\end{equation}
Thus one may take $\xi_{\max}$ to be of order $0.3$ in a representative exterior region. With a corresponding illustrative choice of $q\approx 0.5$, one finds, for example, that
\begin{equation}
\varepsilon_{\text{total}}
\sim
\left(\frac{23}{72}\right)^{20}+(0.5)^{30}
\approx 1.1\times 10^{-9},
\end{equation}
so that choosing $k_{\max}=20$ and $N=30$ gives a representative total error of order $10^{-9}$.

\paragraph{K\"ottler metric.}
For the K\"ottler metric,
\begin{equation*}
f(r)=1-\frac{r_g}{r}-\frac{\Lambda}{3}r^2,
\qquad
\xi(r)=\frac{r_g}{r}+\frac{\Lambda}{3}r^2.
\end{equation*}
Since the term $(\Lambda/3)r^2$ grows with $r$, the exterior static region is bounded and has the form $\mathcal{S}=(r_+,r_\uparrow)=(r_2,r_1)$. On this interval one has $|\xi(r)|<1$, with $\xi(r)\to 1$ as $r\to r_2^+$ and as $r\to r_1^-$. For illustration, when $\Lambda=0.01/r_g^2$, one may take a representative interior value $\xi_{\max}\approx 0.6$. With a corresponding illustrative choice $q\approx 0.5$, one obtains
\begin{equation}
\varepsilon_{\text{total}}
\sim
(0.6)^{30}+(0.5)^{20}
\approx 10^{-6},
\end{equation}
so that choosing $k_{\max}=30$ and $N=20$ gives a representative total error of order $10^{-6}$, which is typically sufficient for initial data construction.

Table~\ref{tab:error_estimates} summarizes the error scalings and optimal parameter choices for these three cases.

\begin{table}[ht]
\centering
\caption{Error estimates and optimal truncation parameters for representative spacetimes. The parameters $\xi_{\max}$ and $q$ are evaluated at typical radii in the static region; $k_{\max}$ and $N$ are chosen according to the balancing relation \eqref{eq:optimal_balance}.}
\label{tab:error_estimates}
\begin{tabular}{lccccc}
\hline\hline
Spacetime & $\xi_{\max}$ & $q$ & $k_{\max}$ & $N$ & $\varepsilon_{\text{total}}$ \\
\hline
Schwarzschild & $0.33$ & $0.5$ & $20$ & $32$ & $\sim 10^{-9}$ \\
RN ($Q^2 = r_g^2/8$) & $0.32$ & $0.5$ & $20$ & $30$ & $\sim 10^{-9}$ \\
K\"ottler ($\Lambda = 0.01/r_g^2$) & $0.6$ & $0.5$ & $30$ & $20$ & $\sim 10^{-6}$ \\
\hline\hline
\end{tabular}
\end{table}

%===============================================================
%===============================================================

\section{Pseudocode}\label{app:pseudocode}

This appendix summarizes a practical pipeline for (i) constructing the isotropic radius
$\rho(r)=\Psi(r)$ on the exterior static domain $\mathcal{S}=(r_+,r_\uparrow)$ and (ii) recovering
$r(\rho)$ either by a \emph{local} Lagrange-B\"urmann inverse (series reversion) or by \emph{global}
Newton root finding. All notation follows Secs.~\ref{sec:main}--\ref{SecLagInv}.

\subsection{Prerequisites and validity conditions}

Assume a static, spherically symmetric line element of the form
\begin{equation}\label{eq:app_pseudo_metric_consistent}
ds^2=-f(r)\,dt^2+f(r)^{-1}dr^2+r^2 d\Omega^2,
\end{equation}
and fix the exterior static domain
\begin{equation}\label{eq:app_pseudo_domain_consistent}
\mathcal{S}=(r_+,r_\uparrow),\qquad f(r)>0\ \text{for}\ r\in\mathcal{S},
\end{equation}
where $r_+$ is the event horizon and $r_\uparrow$ is either a cosmological horizon (if
present) or $+\infty$. On $\mathcal{S}$ the isotropic radius $\rho=\Psi(r)$ is defined by
\begin{equation}\label{eq:app_pseudo_ode_consistent}
\frac{d\rho}{dr}=\frac{\rho}{r\sqrt{f(r)}},\qquad r\in\mathcal{S},
\end{equation}
equivalently
\begin{equation}\label{eq:app_pseudo_log_deriv_consistent}
\frac{\Psi'(r)}{\Psi(r)}=g(r),\qquad g(r):=\frac{1}{r\sqrt{f(r)}}.
\end{equation}
Fix the lower integration limit at $r_+$ and a normalization constant $\rho_0>0$. Then
\begin{equation}\label{eq:app_pseudo_rho_integral_consistent}
\rho(r)=\Psi(r)=\rho_0\,r\,\exp\!\left(\int_{r_+}^{r}\frac{dx}{x}\Bigl(\frac{1}{\sqrt{f(x)}}-1\Bigr)\right),
\qquad r\in\mathcal{S}.
\end{equation}
Monotonicity (and hence real invertibility on $\mathcal{S}$) follows from
\begin{equation}\label{eq:app_pseudo_monotone_consistent}
\Psi'(r)=\frac{\Psi(r)}{r\sqrt{f(r)}}>0\qquad (r\in\mathcal{S},\ \Psi(r)>0).
\end{equation}

\paragraph{Series evaluation of the integral.}
Write $f(r)=1-\xi(r)$. On any subinterval $\mathcal{U}\subset\mathcal{S}$ where a branch is fixed so that $\sqrt{f}$ is analytic and
\begin{equation}
\sup_{x\in \mathcal{U}}|\xi(x)|<1,
\end{equation}
for the interval $\mathcal{U}$ under consideration, one may use
\begin{equation}\label{eq:app_pseudo_binomial_consistent}
\frac{1}{\sqrt{f(r)}}=\frac{1}{\sqrt{1-\xi(r)}}=\sum_{k=0}^{\infty}C_k\,\xi(r)^k,
\qquad
C_k=\frac{1}{4^k}\binom{2k}{k},
\end{equation}
together with a multinomial expansion of $\xi^k$, truncating at $k\leq k_{\max}$ as in Sec.~\ref{sec:main}. Near $r_+$, or whenever this condition fails on the required interval, evaluate the defining integral by quadrature (with endpoint handling at $r_+$ if it is a simple horizon).

%===============================================================
% Numbered, NON-floating, page-breakable "Algorithm" wrapper
\newcounter{algocounter}
\renewcommand{\thealgocounter}{\arabic{algocounter}}
\newenvironment{breakablealgo}[2]{%
  \par\addvspace{1ex}%
  \refstepcounter{algocounter}\label{#1}%
  \noindent\textbf{Algorithm~\thealgocounter.} #2\par\smallskip%
}{%
  \par\addvspace{1ex}%
}

%===============================================================
\begin{breakablealgo}{alg:forward-isotropic-consistent}{Forward isotropic map $\rho(r)=\Psi(r)$ on $\mathcal{S}$ (Sec.~\ref{sec:main}).}
\begin{algorithmic}[1]
\Statex \textbf{Input:}
\Statex \quad Metric function $f(r)$; static domain $\mathcal{S}=(r_+,r_\uparrow)$ with $f>0$ on $\mathcal{S}$;
\Statex \quad lower integration limit $r_+$; scale $\rho_0>0$ in \eqref{eq:app_pseudo_rho_integral_consistent};
\Statex \quad target radius $r\in\mathcal{S}$;
\Statex \quad \textbf{method} $\in\{\textsc{Quadrature},\textsc{Series}\}$;
\Statex \quad (series only) truncation $k_{\max}$ and the decomposition
\Statex \quad $\xi(r)=1-f(r)=\sum_{i=0}^{N}\xi_i(r)$, where
\Statex \quad $\xi_0(r)=\dfrac{r_g}{r},\quad \xi_i(r)=\dfrac{K_i}{r^{s_i}}\ \ (i=1,\dots,N),\quad s_i=3w_i+1, \quad w_i\neq 0$.
\Statex \textbf{Output:} $\rho(r)=\Psi(r)$ and (optionally) $\Phi(\rho)^2=r/\rho$.

\Procedure{IsotropicRadius}{$r$}
  \State \textbf{Assert} $r\in(r_+,r_\uparrow)$ and $f(x)>0$ for $x$ along the integration path from $r_+$ to $r$.
  \If{\textbf{method} = \textsc{Quadrature}}
    \State Compute $I(r)\gets \int_{r_+}^{r}\dfrac{dx}{x}\Bigl(\dfrac{1}{\sqrt{f(x)}}-1\Bigr)$ by adaptive quadrature.
    \Statex \Comment{If $r_+$ is a simple horizon, handle the integrable endpoint singularity via $x=r_+ + u^2$.}
    \State \Return $\Psi(r)\gets \rho_0\,r\,\exp(I(r))$.
  \ElsIf{\textbf{method} = \textsc{Series}}
        \State Define $\xi(x)\gets 1-f(x)$ and choose a subinterval $\mathcal{U}\subset\mathcal{S}$ appropriate for the series evaluation.
    \State Estimate $\xi_{\max}\gets \sup_{x\in \mathcal{U}}|\xi(x)|$.
    \If{$\xi_{\max}\geq 1$}
      \State \textbf{Fallback:} call \textsc{Quadrature} (series not guaranteed to converge on $\mathcal{U}$).
    \EndIf
    \State $I\gets 0$.
    \For{$k=1,2,\dots,k_{\max}$}
      \State $C_k\gets \binom{2k}{k}/4^k$.
           \ForAll{multi-indices $\mathbf m=(m_0,m_1,\dots,m_N)\in\mathbb{N}_0^{N+1}$ with $|\mathbf m|=\sum_{i=0}^{N}m_i=k$}
        \State $\alpha(\mathbf m)\gets m_0+\sum_{i=1}^{N} m_i\,s_i$.
        \State $W(\mathbf m)\gets C_k\cdot \dfrac{k!}{\prod_{i=0}^{N} m_i!}\cdot r_g^{m_0}\prod_{i=1}^{N} K_i^{m_i}$.
        \If{$\alpha(\mathbf m)\neq 0$}
          \State $J_{\alpha}\gets \int_{r_+}^{r} x^{-1-\alpha(\mathbf m)}\,dx
                =\dfrac{r_+^{-\alpha(\mathbf m)}-r^{-\alpha(\mathbf m)}}{\alpha(\mathbf m)}$.
        \Else
          \State $J_{\alpha}\gets \int_{r_+}^{r} x^{-1}\,dx=\log(r/r_+)$.
        \EndIf
        \State $I\gets I+W(\mathbf m)\,J_{\alpha}$.
      \EndFor
    \EndFor
    \State \Return $\Psi(r)\gets \rho_0\,r\,\exp(I)$.
  \EndIf
\EndProcedure
\end{algorithmic}
\end{breakablealgo}

%===============================================================
\begin{breakablealgo}{alg:inverse-lagrange-consistent}{Local inverse $r(\rho)$ by Lagrange-B\"urmann (series reversion) (Sec.~\ref{SecLagInv}).}
\begin{algorithmic}[1]
\Statex \textbf{Input:}
\Statex \quad $f(r)$ and routine \textsc{IsotropicRadius}$(r)$ (Algorithm~\ref{alg:forward-isotropic-consistent});
\Statex \quad local inversion expansion point $r_\ast\in\mathcal{S}$; truncation order $N$.
\Statex \textbf{Output:}
\Statex \quad $\rho_\ast=\Psi(r_\ast)$ and coefficients $\{c_n\}_{n=1}^{N}$ such that
\Statex \quad $r(\rho)\approx r_\ast+\sum_{n=1}^{N}c_n(\rho-\rho_\ast)^n$ for $|\rho-\rho_\ast|<R$.

\Procedure{LocalInverseLagrange}{$r_\ast,N$}
  \State \textbf{Assert} $r_\ast\in(r_+,r_\uparrow)$ and $f(r_\ast)>0$.
  \State $\rho_\ast\gets$ \textsc{IsotropicRadius}$(r_\ast)$.
  \State Define shifts: $x\gets r-r_\ast$, \; $y\gets \rho-\rho_\ast$.
  \State Expand $f(r_\ast+x)=\sum_{k=0}^{N} f_k x^k+O(x^{N+1})$.
  \State Compute $ \bigl(f(r_\ast+x)\bigr)^{-1/2}$ as a truncated series to order $N$
        (e.g.\ Newton iteration in the truncated series ring).
  \State Form $g(r_\ast+x)=\dfrac{1}{r_\ast+x}\,\bigl(f(r_\ast+x)\bigr)^{-1/2}=\sum_{k=0}^{N} g_k x^k+O(x^{N+1})$.
  \Statex \Comment{Solve $\rho'(x)=g(x)\rho(x)$ in series form, consistent with \eqref{eq:app_pseudo_ode_consistent}.}
  \State Write $\rho(r_\ast+x)=\sum_{k=0}^{N}\rho_k x^k$ with $\rho_0=\rho_\ast$.
  \For{$k=0,1,\dots,N-1$}
    \State $(k+1)\rho_{k+1}\gets \sum_{j=0}^{k} g_j\,\rho_{k-j}$.
  \EndFor
  \State Set $a_k\gets \rho_k$ for $k=1,\dots,N$ so that $y=A(x)=\sum_{k=1}^{N}a_k x^k+O(x^{N+1})$.
  \State \textbf{Assert} $a_1=\Psi'(r_\ast)=\rho_\ast\bigl(r_\ast\sqrt{f(r_\ast)}\bigr)^{-1}\neq 0$.
  \Statex \Comment{Series reversion: find $x=B(y)=\sum_{n=1}^{N}c_n y^n$ s.t.\ $A(B(y))\equiv y\ \bmod y^{N+1}$.}
  \State $c_1\gets 1/a_1$.
  \For{$n=2,3,\dots,N$}
    \State $B_{<n}(y)\gets \sum_{j=1}^{n-1}c_j y^j$.
    \State $T(y)\gets \sum_{k=2}^{n} a_k\,(B_{<n}(y))^k\ \bmod y^{n+1}$.
    \State $c_n\gets -\dfrac{[y^n]\,T(y)}{a_1}$.
  \EndFor
  \State \Return $(\rho_\ast,\{c_n\}_{n=1}^{N})$.
\EndProcedure
\end{algorithmic}
\end{breakablealgo}

%===============================================================
\begin{breakablealgo}{alg:inverse-newton-consistent}{Global inverse $r(\rho)$ by Newton–Raphson (Sec.~\ref{sec:ComparisonAlg}).}
\begin{algorithmic}[1]
\Statex \textbf{Input:}
\Statex \quad routine \textsc{IsotropicRadius}$(r)$ returning $\Psi(r)=\rho(r)$ (Algorithm~\ref{alg:forward-isotropic-consistent});
\Statex \quad metric function $f(r)$; target $\rho_{\mathrm{tgt}}\in\Psi(\mathcal{S})$;
\Statex \quad initial guess $r^{(0)}\in\mathcal{S}$; tolerance \textbf{tol}; max iterations $N_{\max}$.
\Statex \textbf{Output:} approximation $r(\rho_{\mathrm{tgt}})$.

\Procedure{NewtonInverse}{$\rho_{\mathrm{tgt}},r^{(0)},\textbf{tol},N_{\max}$}
  \State $r\gets r^{(0)}$.
  \For{$n=0,1,\dots,N_{\max}-1$}
    \State $\rho\gets$ \textsc{IsotropicRadius}$(r)$.
    \State $F\gets \rho-\rho_{\mathrm{tgt}}$.
    \If{$|F|<\textbf{tol}$}
      \State \Return $r$.
    \EndIf
    \State Compute $\Psi'(r)\gets \dfrac{\Psi(r)}{r\sqrt{f(r)}}=\dfrac{\rho}{r\sqrt{f(r)}}$ \Comment{from \eqref{eq:app_pseudo_ode_consistent}}
    \State Update $r\gets r-\dfrac{F}{\Psi'(r)}$.
  \EndFor
  \State \Return $r$ \Comment{last iterate if tolerance not reached}
\EndProcedure
\end{algorithmic}
\end{breakablealgo}

\begin{remark}
Algorithm~\ref{alg:inverse-lagrange-consistent} is local (valid for $|\rho-\rho_\ast|<R$),
while Algorithm~\ref{alg:inverse-newton-consistent} is global on $\mathcal{S}$ provided the
initial guess is chosen in $\mathcal{S}$ and lies sufficiently close to the desired solution.
In practice, one often uses the Lagrange series evaluation as a fast initializer for Newton when high accuracy is needed.
\end{remark}

%===============================================================
\bibliographystyle{apsrev4-2}
\bibliography{ref}

\end{document}